%% file: paper.tex
\documentclass[runningheads]{llncs}
\newcommand{\iso}{\simeq}

\usepackage{graphicx}
\usepackage[conf={no link to proof,restate,text proof translated=}]{proof-at-the-end}
\input{prelude}

\newbox\inlinefigbox

\newcommand{\citet}[1]{\cite{#1}}

\renewcommand{\mp}{\textit{mp}}

\newcommand{\TF}{\textbf{TF}}

\newcommand{\aff}{\textit{aff}}

\newcommand{\spec}{\textit{spec}}
\newcommand{\dom}{\text{dom}}

\newcommand{\Hom}{\textsf{H}}
\newcommand{\Det}{\textsf{Det}}
\newcommand{\DATS}{\mathbf{DATS}}
\newcommand{\QDATS}{\mathbb{Q}\text{-}\mathbf{DATS}}

\newcommand{\DLTS}{\mathbf{DLTS}}
\newcommand{\QDLTS}{\mathbb{Q}\text{-}\mathbf{DLTS}}
\newcommand{\tr}[3]{#2 \rightarrow_{#1} #3}

\newcommand{\DTA}{\textsf{DTA}}
\newcommand{\Homogenize}{\textsf{homog}}
\newcommand{\Int}{\mathbb{Z}}
\newcommand{\IntR}[1]{#1|_{\mathbb{Z}}}

\begin{document}

\title{Reflections on Termination of Linear Loops}

%
 \author{Shaowei Zhu
 \and
 Zachary Kincaid}
 \institute{Princeton University, Princeton NJ 08544, USA\\
 \email{\{shaoweiz,zkincaid\}@cs.princeton.edu}}
%
\maketitle              
\begin{abstract}

\input{abstract}

\end{abstract}

\input{introduction}
\input{background}

\input{bestabs}
\input{dta}

\input{cta}
\input{evaluation}

\input{related}


\bibliographystyle{splncs04}
\bibliography{references}

\input{appendix}

\end{document}

%% file: prelude.tex
\usepackage{amsmath,stmaryrd,amsfonts,amssymb}
\usepackage{thmtools}
\usepackage{thm-restate}
\usepackage{subcaption}

\usepackage{tikz}
\usepackage{pgfplots}
\usetikzlibrary{calc,matrix,arrows,automata,positioning,chains}

\usepackage{tikz-cd}
\pgfplotsset{compat=1.14}

\newcommand{\tuple}[1]{\left\langle #1 \right\rangle}     
\newcommand{\set}[1]{\left\{ #1 \right\}}                 
\newcommand{\true}{\textit{true}}

\newcommand{\defeq}{\triangleq}
\renewcommand{\vec}[1]{\mathbf{#1}}

\newcommand{\transpose}[1]{#1^\intercal}

\renewcommand{\sim}[3]{#1 : #2 \rightarrow #3}

\definecolor{black}{RGB}{0,0,0}
\definecolor{green}{RGB}{25,160,80}
\definecolor{red}{RGB}{192,57,43}
\definecolor{blue}{RGB}{41,128,185}
\definecolor{orange}{RGB}{255,99,0}
\definecolor{gray}{RGB}{75,75,75}
\definecolor{lightgrey}{RGB}{240,240,240}
\definecolor{purple}{RGB}{155, 89, 182}
\definecolor{darkgrey}{RGB}{52, 73, 94}

\usepackage[noend,linesnumbered]{algorithm2e}
\SetKwInOut{Input}{Input}
\SetKwInOut{Output}{Output}
\SetKwRepeat{Do}{do}{while}

\SetCommentSty{mycommfont}
\SetKwProg{Fn}{Subroutine}{ begin}{end}
\SetKwProg{Alg}{Algorithm}{ begin}{end}
\SetKwComment{Inv}{\rm\textbf{invariant:} }{}

\usepackage{listings}
\lstdefinestyle{base}{
  language=C,
  emptylines=1,
  breaklines=true,
  basicstyle=\color{black}\fontfamily{pzc}\selectfont\tt,
  commentstyle=\color{gray}\rm\itshape,
  keywordstyle=\rm\bfseries,
  identifierstyle=\rm\itshape,
  escapechar=@,
  morekeywords=[1]{then,do,halt},
  morekeywords=[2]{assert},
  morekeywords=[3]{requires,ensures},
  keywordstyle	= [2]\color{red}\bf,
  keywordstyle	= [3]\bf\color{gray},
  numberstyle=\footnotesize\color{gray},
  moredelim=**[is][\bf\color{green}]{@!}{!@},
  moredelim=**[is][\bf\color{red}]{@?}{?@},
  literate=
    {<=}{{$\leq$}}1
    {>=}{{$\geq$}}1
    {!}{{$\neg$}}1
    {!=}{{$\neq$}}1
    {||}{{$\lor$}}1
    {&&}{{$\land$}}1
    {->}{{$\rightarrow$}}1
    {_1}{$_{1}$}2
    {_2}{$_{2}$}2
    {_3}{$_{3}$}2
}
\newcommand{\cinline}[1]{\mbox{\lstinline[style=base,mathescape]{#1}}}

\usepackage{mathpartir} 
\usepackage{booktabs}   
\usepackage{hyperref}   
\usepackage{cleveref}   
\usepackage{thmtools} 
\usepackage{thm-restate}
\usepackage{multirow}
\usepackage{adjustbox}

\DeclareFontFamily{U}  {MnSymbolC}{}
\DeclareFontShape{U}{MnSymbolC}{m}{n}{
    <-6>  MnSymbolC5
   <6-7>  MnSymbolC6
   <7-8>  MnSymbolC7
   <8-9>  MnSymbolC8
   <9-10> MnSymbolC9
  <10-12> MnSymbolC10
  <12->   MnSymbolC12}{}
\DeclareSymbolFont{MnSyC}{U}{MnSymbolC}{m}{n}
\DeclareMathSymbol{\righthalfcup}{\mathrel}{MnSyC}{184}

\usepackage{qed}
\newenvironment{mexample}[1][]{%
   \begin{example}%
   \pushQED{\let\qedsymbol\romanqed\qed}%
 }{%
   \popQED%
 \end{example}%
}

\newenvironment{mproofEnd}[1][]{%
   \begin{proofEnd}%
 }{%
 \end{proofEnd}%
}

%% file: abstract.tex
This paper shows how techniques for linear dynamical systems can be used to reason about the behavior of general loops.  We present two main results.  First, we show that every loop that can be expressed as a transition formula in linear integer arithmetic has a \textit{best} model as a \textit{deterministic affine transition system}.  Second, we show that for any linear dynamical system $f$ with integer eigenvalues and any integer arithmetic formula $G$, there is a linear integer arithmetic formula that holds exactly for the states of $f$ for which $G$ is eventually invariant.  Combining the two, we develop a monotone conditional termination analysis for general loops.

%% file: introduction.tex
\section{Introduction}

Linear and affine dynamical systems are a model of computation that is easy to analyze (relative to non-linear systems), making them useful across a broad array of applications.  In the context of program analysis, affine dynamical systems correspond to loops of the form

\setbox\inlinefigbox=\hbox\bgroup
\begin{lstlisting}[style=base]
while (G(@$\vec{x}$@)) do  @$\vec{x}$@ := A@$\vec{x}$@ + @$\vec{b}$@
\end{lstlisting}
\egroup
\hbox to \hsize{\hfil \box \inlinefigbox \hss ($\dagger$)}

\noindent where $G$ is a formula, $A$ is a matrix, $\vec{x}$ is a vector of program variables, and $\vec{b}$ is a constant vector.  The termination problem for such loops has been shown to be decidable  for several variations of this model \cite{CAV:Tiwari2004,CAV:Braverman2006,SODA:OPW2015,POPL:CBKR2019,CAV:FG2019}.
However, few loops in real programs take this form, and so this work has not yet made an impact on practical termination analysis tools.  This paper bridges the gap between theory and practice, showing how techniques for linear and affine dynamical systems can be used to reason about general programs.


\begin{mexample}
We illustrate our methodology using the example program in Figure~\ref{fig:lds} (left).   First, observe
that although the body of this loop is not of the form ($\dagger$),
the value of the sum $x+y$ decreases by $z$ each iteration, and $z$
remains the same.  Thus, we can approximate the loop by the linear
dynamical system in Figure~\ref{fig:lds} (right), where the nature of
the approximation is given by the linear map in the center of
Figure~\ref{fig:lds} (i.e., the $a$ coordinate corresponds to $x+y$, and the $b$ coordinate to $z$).  The linear map is a simulation, in the sense
that it transforms the state space of the program into the state space
of the linear dynamical system so that every step in the loop has a
corresponding step in the linear dynamical system.

Next, we compute the image of the guard of the loop ($x \geq 0 \land y \geq 0$) under the 
simulation, which yields $a \geq 0$ (corresponding to the constraint $x + y \geq 0$ over the original program variables).  We can compute a closed form for this
constraint holding on the $k$th iteration of the loop by exponentiating the dynamics matrix of the linear dynamical system, multiplying on the
left by the row vector corresponding to the constraint, and on the
right by the simulation:
\[\underbrace{\begin{bmatrix}1 & 0\end{bmatrix}}_{\text{Constraint}}\underbrace{\begin{bmatrix}
      1 & -1\\
      0 & 1
    \end{bmatrix}^k}_{\text{Dynamics}}\underbrace{\begin{bmatrix} 0 & 1 & 1 & 0\\ 0 & 0 & 0 & 1\end{bmatrix}}_{\text{Simulation}}\begin{bmatrix}w\\x\\y\\z\end{bmatrix} = (x+y) - kz. \]
  We then analyze the asymptotic behavior of the closed form:
  \[ \text{As } k \rightarrow \infty, (x+y) - kz \rightarrow \begin{cases}
    -\infty & \text{if } z > 0\\
    x+y & \text{if } z = 0\\
    \infty & \text{if } z < 0
  \end{cases} \]
  We conclude that $z > 0 \lor (x+y) < 0$ is a sufficient condition for the loop to terminate.
\end{mexample}

\begin{figure}[t]
\centering
  \begin{tikzpicture}
    \node (program) {\begin{minipage}{5.5cm}
\begin{lstlisting}[style=base,numbers=left,xleftmargin=2em]
z := 1    
while (x >= 0 && y >= 0) do
  w := 3w + x + 1
  if ((x - y) % 2 == 0):
    x := x - z
   else:
    y := y - z
\end{lstlisting}
  \end{minipage}};

  \node [right of=program, node distance=8cm] (lds) {$\begin{bmatrix}a\\b\end{bmatrix} := \begin{bmatrix}
      1 & -1\\
      0 & 1
    \end{bmatrix}
      \begin{bmatrix}a\\b\end{bmatrix}$};
  \draw (program) edge[->,dashed,thick] node[above]{$  \begin{bmatrix}
    a\\
    b
  \end{bmatrix} = \begin{bmatrix}
    0 & 1 & 1 & 0\\
    0 & 0 & 0 & 1
  \end{bmatrix}\begin{bmatrix}
    w\\
    x\\
    y\\
    z
  \end{bmatrix}$} (lds);
  \end{tikzpicture}
  \caption{Over-approximation of a loop by a linear dynamical system. \label{fig:lds}}
\end{figure}
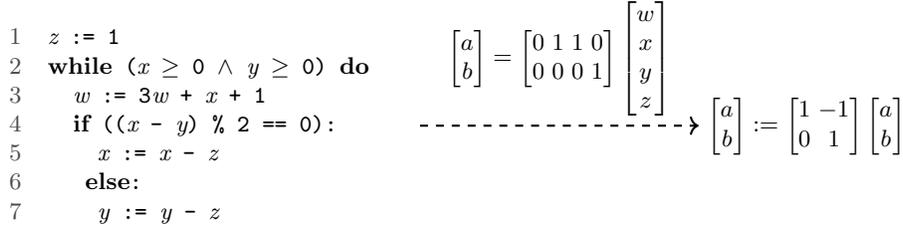

The paper is organized as follows.
To serve as the class of ``linear models'' of loops, we introduce \textit{deterministic affine transition system}s (DATS), a computational model that generalizes affine dynamical systems.  Section~\ref{sec:linear-abs} shows that any loop expressed as a linear integer arithmetic formula has a \textit{DATS-reflection}, which is a best representation of the behavior of the loop as a DATS.  Moreover, this holds for a restricted class of DATS with rational eigenvalues. Section~\ref{sec:dta} shows that for a linear map $f$ with integer eigenvalues and a linear integer arithmetic formula $G$, there is a linear integer arithmetic formula that holds exactly for those states $x$ such that $G(f^k(x))$ holds for all but finitely many $k \in \mathbb{N}$.  Section~\ref{sec:summary} brings the results together, showing that the analysis of a DATS with rational eigenvalues can be reduced to the analysis of a linear dynamical system with integer eigenvalues.  The fact that DATS-reflections are \textit{best} implies monotonicity of the analysis.  Finally, in Section~\ref{sec:evaluation}, we demonstrate experimentally that the analysis can be successfully applied to general programs, using the framework of algebraic termination analysis \cite{arxiv} to lift our loop analysis to a whole-program conditional termination analysis.  Some proofs are omitted for space, but may be found in the appendix.

%% file: background.tex
\section{Preliminaries}

This paper assumes familiarity with linear algebra -- see for example \cite{Lax2007}.  We recall some basic definitions below.

In the following, a \textbf{linear space} refers to a finite-dimensional linear space over the field of rational numbers $\mathbb{Q}$.
For $V$ a linear space and $U \subseteq V$, $\textit{span}(U)$ is the linear space generated by $U$; i.e., the smallest linear subspace of $V$ that contains $U$.  An \textbf{affine subspace} of a linear space $V$ is the image of a linear subspace of $V$ under a translation (i.e., a set of the form $\{ v + v_0 : v \in U \}$ for some linear subspace $U \subseteq V$ and some $v_0 \in V$).  For any scalar $a \in \mathbb{Q}$, and any linear space $V$, we use $\underline{a}$ to denote the linear map $\underline{a} : V \rightarrow V$ that maps $v \mapsto av$ (in particular, $\underline{1}$ is the identity).
A \textbf{linear functional} on a linear space $V$ is a linear map $V \rightarrow \mathbb{Q}$; the set of all linear functionals on $V$ forms a linear space called the \textbf{dual space} of $V$, denoted $V^\star$.  A linear map $f : V_1 \rightarrow V_2$ induces a dual linear map $f^\star : V_2^\star \rightarrow V_1^\star$ where $f^\star(g) \defeq g \circ f$.  For any linear space $V$, $V$ is naturally isomorphic to $V^\star{}^\star$, where the isomorphism  maps $x \mapsto \lambda f : V^\star. f(x)$.

Let $V$ be a linear space.  A linear map $f : V \rightarrow V$ is associated with a \textbf{characteristic polynomial} $p_f(x)$, which is defined to be the determinant of $(xI - A_f)$, where $A_f$ is a matrix representation of $f$ with respect to some basis (the choice of which is irrelevant).  Define the \textbf{spectrum} (set of eigenvalues) of $f$ to be the set of (possibly complex) roots of its characteristic polynomial, $\spec(f) \defeq \{ \lambda \in \mathbb{C} : p_f(\lambda) = 0 \}$.  We say that $f$ has \textbf{rational spectrum} if $\spec(f) \subseteq \mathbb{Q}$; equivalently (by the spectral theorem -- see e.g. \cite[Ch. 6, Theorem 7]{Lax2007}):
\begin{itemize}
    \item There is a basis $\{x_1,...,x_n\}$ for $V$ consisting of \textit{generalized (right) eigenvectors}, satisfying $(f - \underline{\lambda_i})^{r_i}(x_i) = 0$ for some $\lambda_i \in \spec(f)$ and some $r_i \geq 1$ ($r_i$ is called the \textit{rank} of $x_i$)
    \item There is a basis $\{g_1,...,g_n\}$ for $V^\star$ consisting of \textit{generalized left eigenvectors}, satisfying $g_i \circ (f - \underline{\lambda_i})^{r_i} = 0$ for some $\lambda_i \in \spec(f)$ and some $r_i \geq 1$
\end{itemize}
It is possible to determine whether a linear map has rational spectrum (and compute the basis of eigenvectors for $V$ and $V^\star$) in polynomial time by computing its characteristic polynomial \cite{TCS:Keller-Gehrig1985}, factoring it \cite{MA:LLL1982},
and checking whether each factor is linear.

The syntax of linear integer arithmetic (LIA) is given
as follows:
\begin{align*}
  x \in \textsf{Variable}\\
  n \in \mathbb{Z}\\
  t \in \textsf{Term} & ::= x  \mid n \mid n \cdot t \mid t_1 + t_2\\
  F \in \textsf{Formula} & ::= t_1 \leq t_2  \mid (n\mid t) \mid F_1 \land F_2 \mid F_1 \lor F_2 \mid \lnot F  \mid \exists x. F \mid \forall x. F
\end{align*}

Let $X \subseteq \textsf{Variable}$ be a set of variables.  A \textbf{valuation}
over $X$ is a map $v : X \rightarrow \mathbb{Z}$.
If $F$ is a formula
whose free variables range over $X$ and $v$ is a valuation over $X$,
then we say that $v$ satisfies $F$ (written $v \models F$) if the
formula $F$ is true when interpreted over the standard model of the
integers, using $v$ to interpret the free variables.  We write $F
\models G$ if every valuation that satisfies $F$ also satisfies $G$.

\subsection{Transition systems}
A \textbf{transition system} $T$ is a pair $T = \tuple{S_T, R_T}$
where $S_T$ is a set of states and $R_T \subseteq S_T \times S_T$ is a
transition relation.  Within this paper, we shall assume that the state space of any transition system is a finite-dimensional linear space (over $\mathbb{Q}$).
We write $\tr{T}{x}{x'}$ to denote that the pair
$\tuple{x,x'}$ belongs to $R_T$.
We define the \textbf{domain} of a transition system $T$, $\dom(T) \defeq
\{ x \in S_T : \exists x'. \tr{T}{x}{x'} \}$, to be the set of states
that have a $T$-successor.  We define the \textbf{$\omega$-domain}
$\dom^\omega(T)$ of $T$ to be the set of states from which there exist
infinite $T$-computations:
\[ \dom^\omega(T) \defeq \set{ x_0 \in S_T : \exists x_1,x_2,... \text{ such that } \tr{T}{x_0}{\tr{T}{x_1}{\tr{T}{x_2}{\dotsi}}} }\ .\]

A \textbf{transition formula} $F(X,X')$ is an LIA formula whose free variables range over a designated finite set of variables $X$ and a set of ``primed copies'' $X' = \{ x' : x \in X\}$.  For example, a transition formula that represents the body of the loop in Figure~\ref{fig:lds} is

\begin{equation}
    \label{eq:tf}
\begin{array}{ll}
&x \geq 0 \land y \geq 0 \land w' = 3w + x + 1 \land z' = z\\
 \land & \left(\begin{array}{ll}
  &((2 \mid x-y) \land x' = x - z \land y' = y)\\
  \lor &(\lnot (2 \mid x-y) \land y' = y - z \land x' = x)
\end{array}\right)
\end{array}
\end{equation}

\noindent
We use $\TF{}$ to denote the set
of transition formulas.  A transition formula $F(X,X')$
defines a transition system where the state space is the set of functions $X \rightarrow \mathbb{Q}$, and where $v
\rightarrow_F v'$ if and only if both (1) $v$ and $v'$ map each $x \in X$ to an integer and (2) $[v,v'] \models F$, where $[v,v']$ denotes the valuation that maps each $x \in X$ to $v(x)$ and each $x' \in S'$ to $v'(x)$.  Defining the state space of $F$ to be $X \rightarrow \mathbb{Q}$ rather than $X \rightarrow \mathbb{Z}$ is a technical convenience ($X \rightarrow \mathbb{Q} \iso \mathbb{Q}^{|X|}$ is a linear space), but does not materially affect the results of this paper since only (integral) valuations are involved in transitions.

Let $T = \tuple{S_T,R_T}$ be a
transition system.  We say that $T$ is:
\begin{itemize}
\item \textbf{linear} if $R_T$ is a linear subspace of $S_T \times S_T$,
\item \textbf{affine} if $R_T$ is an affine subspace of $S_T \times
  S_T$,
\item \textbf{deterministic} if $\tr{T}{x}{x_1'}$ and $\tr{T}{x}{x_2'}$
  implies $x_1' = x_2'$
\item \textbf{total} if for all $x \in S_T$ there exists some
  $x' \in S_T$ with $\tr{T}{x}{x'}$
\end{itemize}
For example, the transition system $T$ with transition relation
\[
R_T \defeq \set{\tuple{\begin{bmatrix}x\\y\end{bmatrix},\begin{bmatrix}x'\\y'\end{bmatrix}} : \begin{bmatrix}1 & 0 \\
0 & 1\\
0 & 0\end{bmatrix}\begin{bmatrix}x'\\y'\end{bmatrix} = \begin{bmatrix}2 & 1\\ 0 & 1 \\ 0 & 1\end{bmatrix}\begin{bmatrix}x\\y\end{bmatrix} + \begin{bmatrix}0\\0\\1\end{bmatrix}}
\]
is deterministic and affine, but not linear or total.  The transition system $U$ with transition relation
\[
R_U \defeq \set{\tuple{\begin{bmatrix}x\\y\end{bmatrix},\begin{bmatrix}x'\\y'\end{bmatrix}} : \begin{bmatrix}1 & 1\end{bmatrix}\begin{bmatrix}x'\\y'\end{bmatrix} = \begin{bmatrix}\frac{1}{2} & \frac{1}{2}\end{bmatrix}\begin{bmatrix}x\\y\end{bmatrix}}
\]
is total, linear (and affine), but not deterministic.
The classical notion of a \textbf{linear dynamical system}---a
transition system where the state evolves according to a linear
map---corresponds to a \textit{total}, \textit{deterministic},
\textit{linear} transition system.  Similarly, an \textbf{affine dynamical system} is a transition system that is total, deterministic, and affine.

For any map $s : X \rightarrow Y$, and any relation $R
\subseteq X \times X$, define the image of $R$ under $s$ to be the
relation $s[R] = \set{ \tuple{s(x),s(x')} : \tuple{x,x'} \in R}$.  For
any relation $R \subseteq Y \times Y$, define the inverse image of
$R$ under $s$ to be the relation $s^{-1}[R] = \set{ \tuple{x,x'} : \tuple{s(x),s(x')} \in R}$.
Let $T = \tuple{S_T,R_T}$ and $U = \tuple{S_U,R_U}$ be transition
systems. We say that a linear map $s : S_T \rightarrow S_U$ is a
\textbf{linear simulation from $T$ to $U$}, and write $\sim{s}{T}{U}$,
if for all $\tr{T}{x}{x'}$, we have $\tr{U}{s(x)}{s(x')}$.  Observe that the following are equivalent: (1) $s$ is a simulation, (2) $s[R_T] \subseteq R_U$, and (3)
$R_T \subseteq s^{-1}[R_U]$.

An example of a simulation between a transition formula and a linear dynamical system is given in Figure~\ref{fig:lds}.  In fact, there are many linear dynamical systems that over-approximate this loop; however, the simulation and linear dynamical system given in Figure~\ref{fig:lds} is its \textit{best abstraction}.

To formalize the meaning of \textit{best abstractions}, it is convenient to use the language of category theory \cite{SAS:Kincaid2018}.
Any class of transition systems defines a category, where the objects
are transitions systems of that class, and the arrows are linear
simulations between them.  We use boldface letters (\textbf{L}inear, \textbf{A}ffine, \textbf{D}eterministic, \textbf{T}otal) to denote
categories of transition systems (e.g., $\mathbf{DATS}$ denotes the
category of \textbf{D}eterministic \textbf{A}ffine \textbf{T}ransition
\textbf{S}ystems).

If $T$ is a transition system and $\mathbf{C}$ is a category of
transition systems, a \textbf{$\mathbf{C}$-abstraction} of $T$ is a
pair $\tuple{U,s}$ consisting of a transition system $U$ belonging to
$\mathbf{C}$ and a linear simulation $s : T \rightarrow U$.  A
\textbf{$\textbf{C}$-reflection} of $T$ is a $\mathbf{C}$-abstraction
that satisfies a universal property among $\textbf{C}$-abstractions of $T$:
for any $\mathbf{C}$-abstraction $\tuple{V,t}$ of $T$ there exists a
unique simulation $\overline{t} : U \rightarrow V$ such that
$\overline{t} \circ s = t$; i.e., the following diagram commutes:
  \begin{center}
  \begin{tikzpicture}[thick]
    \matrix (m) [matrix of math nodes, row sep=2.5em, column sep=3em,
      text height=1.5ex, text depth=0.25ex] { & V\\ T &
      U\\ }; \draw (m-2-1) edge[->] node[below]{$s$} (m-2-2); \draw (m-2-1) edge[->] node[above
      left]{$t$} (m-1-2); \draw (m-2-2) edge[->,dashed]
    node[right]{$\overline{t}$} (m-1-2);
  \end{tikzpicture}
  \end{center}
  If $\mathbf{D}$ is a category of transition systems and $\mathbf{C}$ is a subcategory
  such that every transition system in $\mathbf{D}$ has a $\mathbf{C}$-reflection, we say that
  $\mathbf{C}$ is a \textbf{reflective subcategory} of $\mathbf{D}$.


Our ultimate goal is to bring techniques from linear dynamical systems to bear on transition formulas.  Figure~\ref{fig:lds} gives an example of a program and its linear dynamical system reflection.  Unfortunately, such reflections do not exist for \textit{all} transition formulas, which motivates our investigation of alternative models.
\begin{theoremEnd}[normal]{proposition} \label{prop:no-best-lts}
  The transition formula $x' = x \land x = 0$ has no
  $\mathbf{TDATS}$-reflection.
\end{theoremEnd}
\begin{mproofEnd}
  Let $F$ be the 1-dimensional transition formula $x' = x \land x =
  0$.  For a contradiction, suppose that $\tuple{A,s}$ is a
  $\mathbf{TDATS}$-reflection of $F$.  Since $F$ contains the origin,
  then so must the transition relation of $A$, and so $A$ is linear.
  Next, consider that for any $\lambda \in \mathbb{Q}$, we have the
  simulation $\textit{id} : F \rightarrow A_\lambda$, where
  $\textit{id}$ is the identity function and $A_\lambda =
  \tuple{\mathbb{Q},x \mapsto \lambda x}$.  Since $\tuple{A,s}$ is a
  reflection of $F$, for any $\lambda$, there is some $t_\lambda$ such
  that $\sim{t_\lambda}{A}{A_\lambda}$ and $\textit{id} = t_\lambda
  \circ s$.  Since $t_\lambda$ is a simulation, we have $\lambda
  t_\lambda = A_\lambda \circ t_\lambda = t_\lambda \circ A$.  Since
  $\textit{id} = t_\lambda \circ s$, we must have $t_\lambda$
  non-zero, and so $t_\lambda$ is a left eigenvector of $A$ with
  eigenvalue $\lambda$.  Since this holds for all $\lambda$, $A$ must
  have infinitely many eigenvalues, a contradiction.
\end{mproofEnd}

%% file: bestabs.tex
\section{Linear abstractions of transition formulas} \label{sec:linear-abs}

Proposition~\ref{prop:no-best-lts} shows that not every transition formula has a total deterministic affine reflection.  In the following we show that \textit{totality} is the only barrier: every transition formula has a (computable) $\DATS{}$-reflection.  Moreover, we show that every transition formula has a \textit{rational spectrum} $\DATS{}$ ($\QDATS$)-reflection, a restricted class of $\DATS{}$ that generalizes affine maps $x \mapsto A\vec{x} + \vec{b}$ where $A$ has rational eigenvalues.  The restriction on eigenvalues makes it easier to reason about the termination behavior of $\QDATS$.

In the remainder of this section, we show that every transition
formula has a $\QDATS{}$-reflection by establishing a chain of
reflective subcategories:
\begin{center}
\begin{tikzpicture}[node distance=3.5cm]
  \node (TF) {$\TF$};
  \node [right of=TF] (ATS) {$\mathbf{ATS}$};
  \node [right of=ATS] (DATS) {$\DATS{}$};
  \node [right of=DATS] (QDATS) {$\QDATS{}$};
  \draw (TF) edge[->] node[above]{Lemma~\ref{lem:best-ats}} (ATS);
  \draw (ATS) edge[->] node[above]{Lemma~\ref{lem:determinization}} (DATS);
  \draw (DATS) edge[->] node[above]{Corollary~\ref{thm:bestqdats}} (QDATS);
\end{tikzpicture}
\end{center}
The fact that $\QDATS{}$ is a reflective subcategory of $\mathbf{TF}$
then follows from the fact that a reflective subcategory of a
reflective subcategory is reflective.


\subsection{Affine abstractions of transition formulas} 

Let $F(X,X')$ be a transition formula.  The \textbf{affine
  hull} of $F$, denoted $\aff(F)$, is the smallest affine set $\aff(F)
\subseteq (X \cup X') \rightarrow \mathbb{Q} \simeq (X \rightarrow \mathbb{Q}) \times (X \rightarrow \mathbb{Q})$ that contains all of the models of $F$.  Reps et
al. give an algorithm that can be used to compute $\aff(F)$, by using
an SMT solver to sample a set of generators
\cite{VMCAI:RSY2004}.

\begin{theoremEnd}{lemma} \label{lem:best-ats}
  Let $F(X,X')$ be a transition formula.  The affine hull of
  $F$ (considered as a transition system) is the best affine
  abstraction of $F$ (where the simulation from $F$ to $\aff(F)$ is
  the identity).
\end{theoremEnd}
\begin{mproofEnd}
  Define $A$ to be the transition system whose transition relation is the affine hull
   of the transition relation of $F$.  Clearly, the identity function is a simulation from
  $F$ to $A$ since $R_F \subseteq R_A = \aff(F)$.  Suppose that $U$
  is an affine transition system and that $s : F \rightarrow U$ is a linear simulation.  Then $s$
  is also a linear simulation from $A$ to $U$, since $s^{-1}[R_U]$ is
  affine and contains $R_F$ (and therefore contains
  $\aff(F)$).
\end{mproofEnd}

\begin{mexample} \label{ex:aff}
Consider the example program in Figure~\ref{fig:lds}.  Letting $F$ denote the transition formula corresponding to the program, $\aff(F)$ can be represented as the solutions to the constraints
\begin{equation} \label{eq:aff}
 \begin{bmatrix}
   1 & 0 & 0 & 0\\
   0 & 1 & 1 & 0\\
   0 & 0 & 0 & 1
  \end{bmatrix}
  \begin{bmatrix}
  w'\\
  x'\\
  y'\\
  z'\\
    \end{bmatrix}
    =
    \begin{bmatrix}
   3 & 1 & 0 & 0\\
   0 & 1 & 1 & -1\\
   0 & 0 & 0 & 1
  \end{bmatrix}
    \begin{bmatrix}
        w\\
  x\\
  y\\
  z\\
    \end{bmatrix}
    +
    \begin{bmatrix}
      1\\0\\0\\0
    \end{bmatrix}
\end{equation}
Notice that $\aff(F)$ is 4-dimensional and has a transition relation defined by 3 constraints, and thus is \textit{not} deterministic.  The next step is to find a suitable projection onto a lower-dimensional space so that the resulting transition system is deterministic.
\end{mexample}

\subsection{Reflections via the dual space}

This section presents a key technical tool that will be used in the next two subsections to prove the existence of reflections.  For any transition system $T$, an abstraction $\tuple{U,s}$ of $T$ consisting of a transition system $U$ and a simulation $s : S_T \rightarrow S_U$ induces a subspace of $S_T^\star$, which is the range of the dual map $s^\star$ (i.e., the set of all linear functionals on $S_T$ of the form $g \circ s$ where $g \in S_U^\star$).  The essential idea is we can apply this in reverse: any subspace  $\Lambda$ of $S_T^\star$ induces a transition system $U$ and a simulation $s : T \rightarrow U$ that satisfies a universal property among all abstractions $\tuple{V,v}$ of $T$ where the range of $v^\star$ is contained in $\Lambda$. We will now formalize this idea.

Let $T$ be a transition system, and let $\Lambda$ be a subspace of
$S_T^\star$.  Define $\alpha_\Lambda(T)$ to be the pair $\alpha_\Lambda(T) \defeq
\tuple{U,s}$ consisting of a transition system $U$ and a linear
simulation $s : T \rightarrow U$ where
\begin{itemize}
    \item $s : S_T \rightarrow \Lambda^\star$ sends each $x \in S_T$ to $\lambda f : \Lambda. f(x)$
    \item $S_U \defeq \Lambda^\star$, and $R_U \defeq s[R_T] = \set{\tuple{s(x),s(x')} : \tuple{x,x'} \in R_T}$
\end{itemize}

\begin{theoremEnd}[normal]{lemma}[Dual space simulation] \label{lem:dual-space-simulation}
Let $T$ be a transition system, let $\Lambda$ be a subspace of $S_T^\star$,
and let $\tuple{U,s} = \alpha_\Lambda(T)$.  Suppose that $Z$ is a transition
system and $z : T \rightarrow Z$ is a simulation such that the range of $z^\star$ is contained in $\Lambda$.  Then there exists a unique
simulation $\overline{z} : U \rightarrow Z$ such that $\overline{z}
\circ s = z$.
\end{theoremEnd}
\begin{mproofEnd}
  The high-level intuition is that since the range of $z^\star$ is contained in $\Lambda$, we may consider it to be a map $z^\star : S_Z^\star \rightarrow \Lambda$; dualizing again, we get a map $z^\star{}^\star : \Lambda^\star \rightarrow S_Z^{\star\star}$, whose domain is $S_U$ and codomain is (isomorphic to) $S_Z$.
  
  More formally, let $j : S_Z \rightarrow S_Z^{\star\star}$ be the natural isomorphism between $S_Z$ and $S_Z^{\star\star}$ defined by
  $j(y) \defeq \lambda g : S_Z^\star. g(y)$.  Define $\overline{z} : \Lambda^\star \rightarrow S_Z$ by
  \[ \overline{z}(h) \defeq j^{-1}(\lambda g : S_Z^\star. h(g \circ z))\ . \]
  First we show that $\overline{z} \circ s = z$.  Let $x \in S_Z$.  Then we have
  \begin{align*}
     (\overline{z} \circ s)(x) &= \overline{z}(s(x))\\
     &= j^{-1}(\lambda g : S_Z^\star. (s(x))(g \circ z))\\
     &= j^{-1}(\lambda g : S_Z^\star. (\lambda f : \Lambda. f(x))(g \circ z))\\
     &= j^{-1}(\lambda g : S_Z^\star. g(z(x)))\\
     &= z(x)\ .
  \end{align*}
  Next we show that $\overline{z}$ is a simulation.  Suppose $\tr{U}{y}{y'}$.  Since $R_{U} = s[R_T]$, there is some $x,x' \in S_T$ such that $\tr{T}{x}{x'}$,
  $s(x) = y$, and $s(x') = y'$.  Since $z : T \rightarrow Z$ is a simulation, we have that $\tr{Z}{z(x)}{z(x)}$, and so
  $\tr{Z}{\overline{z}(s(x))}{\overline{z}(s(x'))}$, and we may conclude that $\tr{Z}{\overline{z}(y)}{\overline{z}(y')}$.
  
  Finally, observe that $s$ is surjective, and therefore the solution to the equation
  $\overline{z} \circ s = z$ is unique.
\end{mproofEnd}

We conclude this section by illustrating how to compute the function $\alpha$ for affine transition systems.  Suppose that $T$ is an affine transition
system of dimension $n$.  We can represent states in $S_T$ by vectors
in $\mathbb{Q}^n$, and the transition relation $R_T$ by a finite set
of transitions $B \subseteq \mathbb{Q}^n \times \mathbb{Q}^n$ that
generates $R_T$ (i.e., $R_T = \aff(B)$).  Suppose that $\Lambda$ is an
$m$-dimensional subspace of $S_T^\star$; elements of $S_T^\star$ can
be represented by $n$-dimensional row vectors, and $\Lambda$ can be
represented by a basis $\transpose{\vec{f}}_1, \dots,
\transpose{\vec{f}}_m$.  We can compute a representation of
$\tuple{U,s} = \alpha_\Lambda(T)$ as follows.  The elements of $S_U =
\Lambda^\star$ can be represented by $m$-dimensional vectors (with respect to the
basis $g_1, \dots, g_m$ such that $g_i$ is the linear map that sends
 $\transpose{\vec{f}}_j$ to 1 if $i=j$ and to 0 otherwise).  The
simulation $s$ can be represented by the $m \times n$ matrix where the
$i$th row is $\transpose{\vec{f}_i}$.  Finally, the transition
relation $R_U$ can be represented by a set of generators $\set{
  \tuple{s(\vec{x}), s(\vec{x}')} : \tuple{\vec{x},\vec{x}'} \in B}$.

\subsection{Determinization}
In this section, we show that any transition system operating over a
finite-dimensional vector space has a best deterministic abstraction,
and give an algorithm for computing the best deterministic affine
abstraction (or \textit{determinization}) of an affine transition
system.

Towards an application of Lemma~\ref{lem:dual-space-simulation}, we
seek to characterize the determinization of a transition system by a
space of functionals on its state space.  For any linear space $V$ and space of functionals $\Lambda$ on $V$, define an equivalence relation $\equiv_\Lambda$ on $V$ by $x \equiv_\Lambda y$ iff $f(x) = f(y)$ for all $f \in \Lambda$.  If $T$ is a transition
system and $\Lambda,\Lambda'$ are spaces of functionals on $S_T$,
we say that $T$ is \textbf{$(\Lambda,\Lambda')$-deterministic} if for all
$x_1,x_2$ $x_1',x_2'$ such that $x_1 \equiv_\Lambda x_2$, $\tr{T}{x_1}{x_1'}$, and
$\tr{T}{x_2}{x_2'}$, then we also have $x_1' \equiv_{\Lambda'} x_1'$.  Observe that if
$D$ is a deterministic transition system and $\sim{d}{T}{D}$ is a simulation, then
$T$ must be $(\Lambda_d,\Lambda_d)$-deterministic, where $\Lambda_d$ is the range of the dual map $d^\star$.

For any $T$ and $\Lambda$, define $\Det(T,\Lambda) \defeq \set{ f : T \text{ is } (\Lambda,\set{f})\text{-deterministic} }$ to be
the greatest set of functionals such that $T$ is $(\Lambda,\Det(T,\Lambda))$-deterministic.
Observe that $\Det(T, -)$ is a monotone operator on the complete
lattice of linear subspaces of $S_T^\star$ (i.e., if $\Lambda_1 \subseteq \Lambda_2$
then $\Det(T,\Lambda_1) \subseteq \Det(T,\Lambda_2)$, since
$\Lambda_1$ induces a coarser equivalence relation than $\Lambda_2$).  By the Knaster-Tarski fixpoint
theorem \cite{PJM:Tarski1955}, $\Det(T,-)$ has a greatest fixpoint,
which we denote by $\Det(T)$.  Then we have that
$T$ is $(\Det(T),\Det(T))$-deterministic, and $\Det(T)$ contains every space $\Lambda$ such that
$T$ is $(\Lambda,\Lambda)$-deterministic.

\begin{theoremEnd}[normal]{lemma}[Determinization] \label{lem:determinization}
  For any transition system $T$, $\alpha_{\Det(T)}(T)$ is a deterministic
  reflection of $T$.
\end{theoremEnd}
\begin{mproofEnd}
  Let $\tuple{D,d} \defeq \alpha_{\Det(T)}(T)$.  First, we show that $D$ is deterministic. Suppose that $\tr{D}{y}{y_1'}$ and $\tr{D}{y}{y_2'}$; we must show that $y_1' = y_2'$.  Since $R_D$ is defined to be $d[R_T]$,  there must be $x_1$, $x_2$, $x_1'$, and $x_2'$ in $S_T$ such that
  $\tr{T}{x_1}{x_1'}$, $\tr{T}{x_2}{x_2'}$,
  $d(x_1) = d(x_2) = y$, $d(x_1') = y_1'$, and $d(x_2') = y_2$.  Since 
  $d(x_1) = d(x_2)$, we have $(\lambda f : \Det(T). f(x_1)) = (\lambda f : \Det(T). f(x_2))$, and therefore 
  $x_1 \equiv_{\Det(T)} x_2$.  We thus have
  $x_1' \equiv_{\Det(T,\Det(T))} x_2'$, and since $\Det(T,\Det(T)) = \Det(T)$, we have $y_1' = d(x_1') = d(x_2') = y_2'$.
 
  It remains to show that $\tuple{D,d}$ is a deterministic \textit{reflection} of
  $T$. Suppose that $\tuple{U,u}$ is another deterministic abstraction of $T$. 
  Define $G$ to be the range of $u^\star$.  Since $U$ is deterministic, we must have
  $G \subseteq \Det(T,G)$, and since $\Det(T)$ is the greatest fixpoint of $\Det(T,-)$ we have $G \subseteq \Det(T)$.  By Lemma~\ref{lem:dual-space-simulation}, there is a unique linear simulation $\overline{u} : D \rightarrow U$ such that
  $\overline{u} \circ d = u$.
\end{mproofEnd}

If a transition system $T$ is affine, then its determinization can be
computed in polynomial time.  Fixing a basis for the state space $S_T$
(of some dimension $n$), we can represent the transition relation of
$T$ in the form $R_T = \{ \tuple{\vec{x},\vec{x}'} : A\vec{x}' =
B\vec{x} + \vec{c} \}$ where $A, B \in \mathbb{Q}^{m \times n}$ and
$\vec{c} \in \mathbb{Q}^m$ (for some $m$).  We can represent
functionals on $S_T$ by $n$-dimensional vectors, where the vector
$\mathbf{v} \in \mathbb{Q}^n$ corresponds to the functional that maps
$\vec{u} \mapsto \transpose{\mathbf{v}} \vec{u}$.  A linear space of
functionals $\Lambda$ can be represented by a system of linear equations $\Lambda
= \set{ \vec{x} : M\vec{x} = 0 }$.  The $i$th row
$\transpose{\vec{a}_i}\vec{v} = \transpose{\vec{b}_i}\vec{u} + c_i$,
of the system of equations $A\vec{x}' = B\vec{x} + \vec{c}$ can be read
as ``$T$ is $(\set{\transpose{\vec{b}_i}},\set{\transpose{\vec{a}_i}})$-deterministic.''  Thus, the functionals $\transpose{\vec{f}}$ such that $T$ is $(\Lambda,\set{\transpose{\vec{f}}})$-deterministic
are those that can be written as a linear combination
of the rows of $A$ such that the corresponding linear combination of
the rows of $B$ belongs to $\Lambda$; i.e.,
\[ \Det(\{ \tuple{\vec{x},\vec{x}'} : A\vec{x}' = B\vec{x} + \vec{c} \} , \{ \vec{f} : M\vec{f} = 0 \}) = \set{ \vec{d} :  \exists \vec{y}. M\transpose{B}\vec{y} = 0 \land \transpose{A}\vec{y} = \vec{d} }\ . \]
A representation of $\Det(T,\Lambda)$ can be computed in polynomial time
using Gaussian elimination.  Since the lattice of linear subspaces of
$S_T^\star$ has height $n$, the greatest fixpoint of $\Det(T,-)$ can be
computed in polynomial time.

\begin{mexample}
  Continuing the example from Figure~\ref{fig:lds} and Example~\ref{ex:aff}, we consider the determinization of the affine transition system in Eq~(\ref{eq:aff}).
  The rows of the matrix on the left-hand side correspond to generators for $\Det(\aff(F),{\mathbb{Q}^4}^\star)$:
  \begin{align*}
  \Det(\aff(F),{\mathbb{Q}^4}^\star) &= 
  \textit{span}(\{\begin{bmatrix} 1 & 0 & 0 & 0 \end{bmatrix},
  \begin{bmatrix} 0 & 1 & 1 & 0 \end{bmatrix},
  \begin{bmatrix} 0 & 0 & 0 & 1 \end{bmatrix}\})\\
  \Det(\aff(F),\Det(\aff(F),{\mathbb{Q}^4}^\star)) &=
  \textit{span}(\{\begin{bmatrix} 0 & 1 & 1 & 0 \end{bmatrix},
  \begin{bmatrix} 0 & 0 & 0 & 1 \end{bmatrix}\})
  \end{align*}
  which is the greatest fixpoint $\Det(\aff(F))$.
  Intuitively: after one step of $\aff(F)$, the values
  of $w$, $x+y$, and $z$ are affine functions of the input;
  after two steps $x+y$ and $z$ are affine functions of the input but
  $w$ is not, since the value of $w$ on the second step depends upon the value of $x$ in the first, and $x$ is not an affine function of the input.
  
  This yields the deterministic reflection $\tuple{D,d}$ (also pictured in Figure~\ref{fig:lds}) where
  \begin{center}
      $R_D = \set{\tuple{\begin{bmatrix}a\\b\end{bmatrix},\begin{bmatrix}a'\\b'\end{bmatrix}} : \begin{bmatrix}a'\\b'\end{bmatrix} = \begin{bmatrix}
      1 & -1\\
      0 & 1
    \end{bmatrix}
      \begin{bmatrix}a\\b\end{bmatrix}}$
      \hspace*{1cm}and\hspace*{1cm}
      $d = \begin{bmatrix}
    0 & 1 & 1 & 0\\
    0 & 0 & 0 & 1
  \end{bmatrix}$
  \end{center}
\end{mexample}

\subsection{Rational-spectrum reflections of DATS} \label{sec:qdats}

In this section, we define rational-spectrum $\DATS{}$ and show that
every $\DATS{}$ has a rational-spectrum-reflection.

In the following, it is convenient to work with transition systems
that are linear rather than affine.  We will prove that every
deterministic \textit{linear} transition system has a best abstraction
with rational spectrum.  The result extends to the affine case through
the use of \textit{homogenization}: i.e., we embed a (non-empty) affine
transition system into a linear transition system with one additional
dimension, such that if we fix that dimension to be $1$ then we
recover the affine transition system.  If the transition relation of a
$\DATS{}$ is represented in the form $A\vec{x}' = B\vec{x} + \vec{c}$,
then its homogenization is simply
\[ \begin{bmatrix}
  A & 0\\
  0 & 1
\end{bmatrix}
\begin{bmatrix}
  \vec{x}'\\
  y
\end{bmatrix}
=\begin{bmatrix}
  B & \vec{c}\\
  0 & 1
\end{bmatrix}
\begin{bmatrix}
  \vec{x}\\
  y
\end{bmatrix}\ .
\]
For a $\DATS{}$ $T$,  we use $\Homogenize(T)$ to denote the pair $\tuple{L, h}$, consisting the $\DLTS$ $L$ resulting from homogenization and the affine simulation
$h : T \rightarrow L$ that maps each $\vec{x} \in S_T$ to $\begin{bmatrix}
 \vec{x} \\ 1 
\end{bmatrix}$ (i.e., the affine simulation $h$ formalizes the idea that if we fix the extra dimension $y$ to be 1, we recover the original $\DATS$ $T$).

Let $T$ be a deterministic linear transition system.  Since our goal
is to analyze the asymptotic behavior of $T$, and all long-running
behaviors of $T$ reside entirely within $\dom^\omega(T)$, we are
interested in the structure of $\dom^\omega(T)$ and $T$'s behavior on
this set.  First, we observe that $\dom^\omega(T)$ 
is a linear subspace of $S_T$
and is computable.  For any $k$, let $T^k$ denote the
linear transition system whose transition relation is the $k$-fold
composition of the transition relation of $R$.  Consider the
descending sequence of linear spaces
\[ \dom(T) \supseteq \dom(T^2) \supseteq \dom(T^3) \supseteq \dots \]
(i.e., the set of states from which there are $T$ computations of
length 1, length 2, length 3, \dots).  Since the space $S_T$
is finite dimensional, this sequence must stabilize at some $k$.
Since the states in $\dom(T^k)$ have $T$-computations of any length
and $T$ is deterministic, we have that $\dom(T^k)$ is precisely
$\dom^\omega(T)$.

Since $T$ is total on $\dom^\omega(T)$ and the successor of a state in
$\dom^\omega(T)$ must also belong to $\dom^\omega(T)$, $T$ defines a
linear map $T|_\omega: \dom^\omega(T) \rightarrow \dom^\omega(T)$.  In
this way, we can essentially reduce asymptotic analysis of $\DATS{}$
to asymptotic analysis of linear dynamical systems.  The asymptotic analysis
of linear dynamical systems developed in Sections~\ref{sec:dta} and~\ref{sec:summary}
requires rational eigenvalues; thus we are interested in $\DATS{}$ $T$ such that
$T|_\omega$ has rational eigenvalues.  With this in mind, we define
$\spec(T) = \spec(T|_\omega)$, and say that $T$ \textbf{has rational
  spectrum} if $\spec(T) \subseteq \mathbb{Q}$.  Define $\QDLTS{}$ to
be the subcategory of $\DLTS{}$ with rational spectrum, and $\QDATS{}$
to be the subcategory of $\DATS{}$ whose homogenization lies in
$\QDLTS{}$.

\begin{mexample}
 Consider the $\DLTS{}$ $T$ with
 \[
 R_T \defeq \set{\tuple{\begin{bmatrix}x\\y\\z\end{bmatrix},\begin{bmatrix}x'\\y'\\z'\end{bmatrix}} :
 \begin{bmatrix}
 1 & 0 & 0 \\
 0 & 1 & 0 \\
 0 & 0 & 1 \\
 0 & 0 & 0 
 \end{bmatrix}
 \begin{bmatrix}x'\\y'\\z'\end{bmatrix} =
  \begin{bmatrix}
 2 & 0 & 1 \\
 0 & 2 & 2 \\
 0 & 0 & 3 \\
 1 & -1 & 0 \\
 \end{bmatrix}
 \begin{bmatrix}x\\y\\z\end{bmatrix}}
 \]
 The bottom-most equation corresponds to a constraint that only vectors where the $x$ and $y$ coordinates are equal have successors, so we have:
 \[ \dom(T) = \set{\transpose{\begin{bmatrix}x&y&z\end{bmatrix}} : x = y} \]
 Supposing that the $x$ and $y$ coordinates are equal in some pre-state, they are equal in the post-state exactly when $z=0$, so we have
 \[ \dom(T^2) = \set{\transpose{\begin{bmatrix}x & y & z\end{bmatrix}} : x = y \land z = 0} \]
 It is easy to check that $\dom(T^3) = \dom(T^2)$, and therefore
 $\dom^\omega(T) = \dom(T^2)$.  The vector $\transpose{\begin{bmatrix}1 & 1 & 0\end{bmatrix}}$ is a basis
 for $\dom^\omega(T)$, and the matrix representation of $T|_\omega$ with respect to this basis is $\begin{bmatrix} 2 \end{bmatrix}$ (i.e., $\transpose{\begin{bmatrix}1 & 1 & 0\end{bmatrix}} \rightarrow_T \transpose{\begin{bmatrix}2 & 2 & 0\end{bmatrix}}$).  Thus we can see $\spec(T) = \set{2}$, and $T$ is a $\QDLTS{}$.
\end{mexample}

Towards an application of Lemma~\ref{lem:dual-space-simulation},
define the \textbf{generalized rational eigenspace} of a DLTS $T$ to
be
\[
E_{\mathbb{Q}}(T) \defeq \textit{span}\left(\set{ f \in S_T^\star : \exists \lambda \in \mathbb{Q}, \exists r \in \mathbb{N}^+. f \circ (T|_\omega - \underline{\lambda})^r = 0 }\right).
\]

\begin{theoremEnd}{lemma}\label{lem:LTS-factor}
  Let $T$ be a DLTS, and define
  $\tuple{Q,q} \defeq \alpha_{E_\mathbb{Q}(T)}(T)$.  Then for any
  $\QDLTS{}$ $U$ and any simulation $s : T \rightarrow U$, there is a unique simulation $\overline{s} : Q \rightarrow U$ such that
  $\overline{s} \circ q = s$.
\end{theoremEnd}
\begin{mproofEnd}
By Lemma~\ref{lem:dual-space-simulation} it is sufficient to show that for any $f \in S_U^\star$, we have $f \circ s \in E_{\mathbb{Q}}(T)$.   Suppose that $U$ has dimension $n$, and $\dom^\omega(U)$ has dimension $m$.
Since the set of functionals $f$ satisfying $f \circ s \in E_{\mathbb{Q}}(T)$ is linear, it is sufficient to construct a basis
$f_1,\dots,f_n$ for $S_U^\star$ and show that $f_i \circ s \in E_{\mathbb{Q}}(T)$ for each $i$.

    By the spectral theorem \cite[Ch. 6, Theorem 7]{Lax2007} and the assumption that $U$ has all rational spectrum,
  there is a basis for $\dom^\omega(U)^\star$ consisting 
  functionals $h_1,\dots,h_m$ that
  for each $i$ there is some $\lambda_i \in \mathbb{Q}$ and $r_i \geq 1$ such that $h_i \circ (U|_\omega - \underline{\lambda_i})^{r_i} = 0$.  For each $i$, let $f_i \in S_U^\star$ be an (arbitrary) extension of $h_i$ to $S_U$.  Let $f_{m+1},\dots,f_n \in S_U^\star$ be a basis for the space of functionals that vanish on $\dom^\omega(U)$.  Claim that $f_1,\dots,f_n$ is linearly independent (and is therefore a basis for $S_U'$).  Suppose that $a_1f_1 + \dots + a_nf_n = 0$; we must show $a_i$ = 0 for each $i$.
  For all $x \in \dom^\omega(U)$ we must have
  \begin{align*}
      0 &= (a_1f_1 + \dots + a_nf_n)(x)  & \text{Assumption}\\
      &= (a_1f_1 + \dots + a_mf_m)(x) &  f_{m+1},\dots,f_n \text{ vanish on } \dom^\omega(U)\\
      &= (a_1h_m + \dots + a_mh_m)(x) & h_i \text{ and } f_i \text{ coincide on } \dom^\omega(U)
  \end{align*} and thus $a_1 = \dots = a_m = 0$, since $h_1,\dots,h_m$ are linearly independent.  Since $a_{m+1}f_{m+1} + \dots + a_nf_n = 0$  and $f_{m+1}, \dots, f_n$ are linearly independent, we must have $a_{m+1} = \dots = a_n = 0$.

  Since $s : T \rightarrow U$ is a simulation, $s$ must map elements of $\dom^\omega(T)$ to elements of $\dom^\omega(U)$; we use $s|_\omega : \dom^\omega(T) \rightarrow \dom^\omega(U)$ to denote the restriction of $s$ to $\dom^\omega(T)$.  We have that $s|_\omega$ is a simulation from $T|_\omega$ to $U|_\omega$, which can be expressed by the equation $s|_\omega \circ T|_\omega = U|_\omega \circ s|_\omega$.
  From the construction of the basis $f_1,\dots,f_n$, we have that $f_i \circ (U|_\omega - \underline{\lambda_i})^{r_i} = 0$ (for $i > m$, since $f_i$ vanishes on $\dom^\omega(U)$ we may take $\lambda_i \defeq 0$ and $r_i \defeq 1$).  Pre-composing with $s|_\omega$, we have
   $f_i \circ (U|_\omega - \underline{\lambda_i})^{r_i} \circ s|_\omega = 0$, and by induction on $r_i$, we have
  $f_i \circ s|_\omega \circ (T|_\omega - \underline{\lambda_i})^{r_i} = f_i \circ (U|_\omega - \underline{\lambda_i})^{r_i} \circ s|_\omega = 0$.  We conclude that
  $f_i \circ s \circ (T|_\omega - \underline{\lambda_i})^{r_i} = 0$ and thus $f_i \circ s \in E_{\mathbb{Q}}(T)$.
\end{mproofEnd}

While $\alpha_{E_\mathbb{Q}(T)}(T)$ satisfies a universal property for
$\QDLTS{}$, it does not necessary belong to $\QDLTS{}$ itself because
it need not be deterministic.  However, by iterative interleaving of
Lemma~\ref{lem:LTS-factor} and determinization as shown in
Algorithm~\ref{alg:rational-spec-abstraction}, we arrive at a
$\QDLTS{}$-reflection.
Example~\ref{ex:best-qdlts-abstraction} in the appendix demonstrates how 
we calculate a $\QDLTS{}$-reflection of a particular $\DLTS{}$.

\begin{theoremEnd}{theorem}
  For any 
  deterministic linear transition system,
  Algorithm~\ref{alg:rational-spec-abstraction} computes a $\QDLTS{}$-reflection.
  \label{thm:best-qdlts}
\end{theoremEnd}
\begin{mproofEnd}
  Let $T$ be a $\DLTS{}$.
  Clearly, if Algorithm~\ref{alg:rational-spec-abstraction}
  terminates, then it returns a $\QDLTS{}$ abstraction of $T$.  Since
  each iteration of the loop decreases the dimension the state space $S_U$, the algorithm must terminate.  It remains to show that
  the abstraction is a reflection.

  Suppose that there exists a linear simulation $v$ from $T$ to
  an arbitrary $\QDLTS{}$ $V$.  We show the loop maintains the invariant
  that there is a unique simulation $\overline{v} : U \rightarrow V$
  such that $v = \overline{v} \circ s$.
  The invariant trivially holds when entering the loop. To show that it is maintained by the loop, we suppose that there exists a unique simulation $\overline{v} : U \rightarrow V$
  such that $v = \overline{v} \circ s$, and show that there exists a unique simulation $\overline{v}' : U' \rightarrow V$
  such that $v = \overline{v}' \circ d \circ q \circ s$, where
  $\tuple{Q,q} = \alpha_{E_\mathbb{Q}(U)}(U)$, and
  $\tuple{U',d} = \alpha_{\Det(Q)}(Q)$.
  
  Since $\overline{v}$ is a simulation from $U$ to $V$, by Lemma~\ref{lem:LTS-factor}, 
  there exists a unique simulation
  $\overline{q}$ from $Q$ to $V$ such that
  $\overline{v} = \overline{q} \circ q$.  Since $V$ is deterministic and $\tuple{U',d}$ is a deterministic reflection of $Q$, there exists a unique simulation
  $\overline{v}'$ from $U'$ to $V$ such that
  $\overline{q} = \overline{v}' \circ d$.  We may then verify the desired result:
\begin{align*}
     \overline{v}' \circ d \circ q \circ s &= \overline{q} \circ q \circ s
      = \overline{v} \circ s
\end{align*}
     
Finally, observe that $d \circ q \circ s$ is the composition of surjective maps and therefore surjective.  It follows that $\overline{v}'$ is the \textit{unique} solution to the above equation.
\end{mproofEnd}

\begin{algorithm}
\Input{A $\DLTS$ $T$.}
\Output{$\QDLTS{}$-reflection of $T$}
$U \gets T$\;
$s \gets \lambda x. x$ \tcc*{Invariant: $s$ is a simulation from $T$ to $U$}
\While{$\spec(U|_\omega) \nsubseteq \mathbb{Q}$} {
  $\tuple{Q,q} \gets \alpha_{E_\mathbb{Q}(U)}(U)$ \tcc*{Lemma~\ref{lem:LTS-factor}}
  $\tuple{U,d} \gets \alpha_{\Det(Q)}(Q)$ \tcc*{Lemma~\ref{lem:determinization}}
  $s \gets d \circ q \circ s$\;
}
\Return{$\tuple{U,s}$}
\caption{Computation of a $\QDLTS{}$-reflection of a $\DLTS{}$}
\label{alg:rational-spec-abstraction}
\end{algorithm}

Finally, by homogenization and Theorem~\ref{thm:best-qdlts}, we
conclude with the desired result:
\begin{theoremEnd}[category=homogenization]{corollary}
  \label{thm:bestqdats}
  $\QDATS{}$ is a reflective subcategory of $\DATS{}$.
\end{theoremEnd}
\begin{mproofEnd}
  Let $T$ be a deterministic affine transition system.  If $R_T$ is
  empty, then it is a $\QDATS{}$ and we are done.  Suppose $R_T$ is non-empty.
  By
  Lemma~\ref{thm:best-qdlts}, $\Hom(T)$ has a $\QDLTS{}$-reflection, say $\tuple{\hat{T},s}$.  The intuition behind the argument is that we may ``dehomogenize'' $\hat{T}$ and $s$ to obtain a $\QDATS{}$-reflection $\tuple{Z,z}$ of $T$.
 
  In the following, for any product space $A \times B$, we use $\pi_1$ and $\pi_2$ to denote the first and second projection ($\pi_1(a,b) \defeq a$, $\pi_2(a,b) \defeq b$).
  Let $K$ be the linear transition system where $S_K = \mathbb{Q}$ and
  $R_K = \set{ \tuple{s,s} : s \in \mathbb{Q} }$; intuitively $K$
  captures the dynamics of the constant dimension introduced by the
  homogenization process.  Clearly
  $\sim{\pi_2}{\Hom(T)}{K}$ is a simulation, and $K$ has rational
  spectrum.  Since $\tuple{\hat{T},s}$ is a $\QDLTS{}$-reflection of
  $\Hom(T)$, there is a unique
  simulation $\sim{\overline{\pi_2}}{\hat{T}}{K}$ such that
  $\overline{\pi_2} \circ s = \pi_2$.  Define $Z$ to be the affine transition system where
  \begin{align*}
    S_Z &= \{ x \in S_{\hat{T}} : \overline{\pi_2}(x) = 0 \}\\
    R_Z &= \{ \tuple{x,x'} : \tuple{x + s(0,1), x' + s(0,1)} \in R_{\hat{T}} \}
  \end{align*}
  Observe that for any $x \in S_T$, we have $\overline{\pi_2}(s(x,0)) =
  \pi_2(x,0) = 0$, and so the function $z$ defined by $z(x) \defeq
  s(x,0)$ maps $S_T$ to $S_Z$.  In fact, $z$ is a simulation from $T$
  to $Z$:
  \begin{align*}
    \tr{T}{x}{x'} &\Rightarrow \tr{\Hom(T)}{\tuple{x,1}}{\tuple{x',1}} & \text{Def'n of $\Hom$}\\
    &\Rightarrow \tr{\hat{T}}{s(x,1)}{s(x',1)} & \text{$s$ is a simulation}\\
    &\Rightarrow \tr{\hat{T}}{(s(x,0)+s(0,1))}{(s(x',0)+s(0,1))} & \text{Linearity}\\
    &\Rightarrow \tr{Z}{s(x,0)}{s(x',0)} & \text{Def'n of $R_Z$}\\
    &\Rightarrow \tr{Z}{z(x)}{z(x')} & \text{Def'n of $z$}
  \end{align*}

  Next we show that $\Hom(Z)$ is isomorphic to $\hat{T}$, and therefore has rational spectrum.  Define a bijective linear
  map $e : S_{\Hom(Z)} \rightarrow S_{\hat{T}}$ by $e(x,c) \defeq x +
  cs(0,1)$.  We now show that $\tr{\Hom(Z)}{(x,c)}{(x',c)} \iff \tr{\hat{T}}{e(x,c)}{e(x',c)}$.  Let $x_0,x_0'$ be such that $\tr{Z}{x_0}{x_0'}$.  Unfolding the definitions of $\Hom(Z)$ and $Z$, we have
  {\footnotesize\begin{align}
      \tr{\Hom(Z)}{(x,c)}{(x',c)} &\iff \tr{Z}{(x + (1-c)x_0)}{(x' +
    (1-c)x_0')}\\
    &\iff  \tr{\hat{T}}{(x + (1-c)x_0 + s(0,1))}{(x' + (1-c)x_0' + s(0,1))} \label{eq:tr1}
  \end{align}}
  Since by assumption $\tr{Z}{x_0}{x_0'}$, we have
  \begin{equation} \label{eq:tr2}
    \tr{{\hat{T}}}{(x_0 + s(0,1))}{(x_0' + s(0,1))}\ .
  \end{equation}
  Since $R_{\hat{T}}$ is closed under linear combinations, we may add
  the transition in Eq.~(\ref{eq:tr1}) to the transition
  Eq.~(\ref{eq:tr2}) scaled by a factor of $(c-1)$ to see that
  \[ \tr{\Hom(Z)}{(x,c)}{(x',c)} \iff \tr{\hat{T}}{x + cs(0,1)}{x'+cs(0,1)} \iff \tr{\hat{T}}{e(x,c)}{e(x',c)}\ . \]

  Finally, we show that $\tuple{Z,z}$ is a $\QDATS{}$-reflection of $T$.
  Suppose that $Y \in \QDATS{}$ and $\sim{y}{T}{Y}$ is a simulation.
  Then we have a simulation $\sim{\Hom(y)}{\Hom(T)}{\Hom(Y)}$, and
  since $\tuple{\hat{T},s}$ is a $\QDLTS$-reflection of $\Hom(T)$
  there is a unique simulation
  $\sim{\overline{\Hom(y)}}{\hat{T}}{\Hom(Y)}$ such that
  $\overline{\Hom(y)} \circ s = \Hom(y)$.  Define a function
  $\overline{y} : S_Z \rightarrow S_Y$ by $\overline{y} \defeq \pi_1
  \circ \overline{\Hom(y)}$.  We now show that $y = \overline{y} \circ
  z$ and that $\overline{y}$ is a simulation from $Z$ to $Y$
  (uniqueness follows from the fact that $z$ is surjective).

  First, we show that $y = \overline{y} \circ z$.  Observe that for any $x \in S_T$, we have
  \begin{align*}
    \overline{y}(z(x)) & = \pi_1(\overline{\Hom(y)}(s(x))) & \text{Def'n of $\overline{y}$ and $z$}\\
    &= \pi_1(\Hom(y)(x,0)) & \Hom(y) = \overline{\Hom(y)} \circ s\\
    &= \pi_1(y(x),0) & \text{Def'n of $\Hom(y)$}\\
    &= y(x)
  \end{align*}
  
  Second, we show that $\overline{y}$ is a simulation from $Z$ to $Y$.
  Claim that for any $x \in S_Z$, we have
  $\overline{\Hom(y)}(x + s(0,1)) = (\overline{y}(x),1)$.
  The fact that $\overline{y}$ is
  a simulation follows:
  \begin{align*}
    \tr{Z}{x}{x'}& \Rightarrow \tr{\hat{T}}{(x + s(0,1))}{(x' + s(0,1))} & \text{Def'n of $Z$}\\
    &\Rightarrow \tr{\Hom(Y)}{\overline{\Hom(y)}(x + s(0,1))}{\overline{\Hom(y)}(x' + s(0,1))} & \text{$\overline{\Hom(y)}$ is a simulation}\\
    &\Rightarrow \tr{\Hom(Y)}{(\overline{y}(x), 1)}{(\overline{y}(x'),1)} & \text{Claim}\\
    &\Rightarrow \tr{Y}{\overline{y}(x)}{\overline{y}(x')} & \text{Def'n of $\Hom(Y)$}
  \end{align*}
  
  It remains to prove the claim.  Observe that
  \begin{align*}
    \overline{\Hom(y)}(x + s(0,1)) &= \overline{\Hom(y)}(x) + \overline{\Hom(y)}(s(0,1)) & \text{Linearity}\\
    &=\overline{\Hom(y)}(x) + \Hom(y)(0,1) & \Hom(y) = \overline{\Hom(y)} \circ s\\
    &=\overline{\Hom(y)}(x) + (y(0),1) & \text{Def'n of $\Hom(y)$}\\
    &=\overline{\Hom(y)}(x) + (0,1) & \text{Linearity}\\
    &=(\pi_1(\overline{\Hom(y)}(x)), \pi_2(\overline{\Hom(y)}(x)) + 1)\\
    &=(\overline{y}(x), \pi_2(\overline{\Hom(y)}(x)) + 1) & \text{Def'n of $\overline{y}$}
  \end{align*}
  so it is sufficient to show that $\pi_2(\overline{\Hom(y)}(x)) = 0$.
  Since $s$ is
  surjective, there is some $\tuple{a,b} \in S_{\Hom(T)} = S_T \times \mathbb{Q}$ with $s(a,b) = x$.
  Finally, we may verify
  \begin{align*}
    \pi_2(\overline{\Hom(y)}(x)) &= \pi_2(\overline{\Hom(y)}(s(a,b))) & \text{Def'n of $a,b$}\\
    &= \pi_2(\Hom(y)(a,b))  & \Hom(y) = \overline{\Hom(y)} \circ s \\
    &= \pi_2(y(a),b) & \text{Def'n of $\Hom(y)$}\\
    &= \pi_2(a,b)\\
    &= \overline{\pi_2}(s(a,b)) & \pi_2 = \overline{\pi_2} \circ s\\
    &= \overline{\pi_2}(x) & \text{Def'n of $a,b$}
  \end{align*}
  and since $x \in Z$ we have $\overline{\pi_2}(x) = 0$.
\end{mproofEnd}

%% file: dta.tex
\section{Asymptotic Analysis of Linear Dynamical Systems} \label{sec:dta}

 This section is concerned with analyzing the behavior of loops of the form
 \[ \textbf{while } (G(\vec{x})) \textbf{ do } \vec{x} := A\vec{x} \ , \]
 where the $G(\vec{x})$ is an LIA formula and $A$ is a matrix with integer spectrum.
 Our goal is to capture the asymptotic behavior of iterating the map $A$ on an initial state $\vec{x}_0$ with respect to the formula $G$.
 Specifically, we show that
\begin{theorem}
  \label{thm:periodic-truth-value-of-guard-under-linear-map}
  For any LIA formula $G$ and any matrix $A$ with integer spectrum, there is a periodic sequence of LIA formulas $H_0, H_1,H_2,\dots$ such that for any initial state $\vec{x}_0 \in \mathbb{Q}^n$,
  there exists $K$ such that for any $k > K$,
  $G(A^k \vec{x}_0)$ holds if and only if $H_k(\vec{x}_0)$ does.
\end{theorem}
Recall that an infinite sequence $H_0, H_1, H_2, \ldots$ is \textit{periodic} if
it is of the form
\[ (H_0, H_1, \dots, H_P)^\omega \defeq H_0, H_1, \dots, H_P, H_0, H_1, \dots, H_P, \dots \]
We call the periodic sequence $(H_0, H_1, \dots, H_P)^\omega$ the \textit{characteristic sequence} of 
the guard formula $G$ with respect to dynamics matrix $A$, and denote it by $\chi(G, A)$.  Note that $G(A^k\vec{x}_0)$ holds for all but finitely many $k$ exactly when $\bigwedge_{i=0}^P H_i(\vec{x}_0)$ holds.

In the remainder of this section, we show how to compute characteristic sequences.
Let $G$ be an LIA formula and let $A$ be a matrix with integer spectrum. To begin, we compute a quantifier-formula $G'$ that is equivalent to $G$ (using, for example, Cooper's algorithm~\cite{Cooper1972}).  We define $\chi(G',A)$ by recursion on the structure of $G'$.  For the logical connectives $\land$, $\lor$, and $\lnot$, characteristic sequences are defined pointwise:
\begin{align*}
 \chi(\neg H, A) &\defeq (\neg(\chi(H, A)_0), \neg(\chi(H, A)_1), \ldots)\\
 \chi(H_1 \land H_2, A) &\defeq (\chi(H_1, A)_0 \land \chi(H_2, A)_0, \chi(H_1, A)_1 \land \chi(H_2, A)_1, \ldots)\\
  \chi(H_1 \lor H_2, A) &\defeq (\chi(H_1, A)_0 \lor \chi(H_2, A)_0, \chi(H_1, A)_1 \lor \chi(H_2, A)_1, \ldots)
\end{align*}
It remains to show how $\chi$ acts on atomic formulas, which take the form of inequalities $t_1 \leq t_2$ and divisibility constraints $n \mid t$.
An important fact that we employ in both cases is that for any linear term $\transpose{\vec{c}}\vec{x}$ over the variables $\vec{x}$, we can compute a closed form for $\transpose{\vec{c}}A^k(\vec{x})$ by symbolically exponentiating $A$.
   Since (by assumption) $A$ has integer eigenvalues, this closed form has the form $\frac{1}{Q}(p(\vec{x},k))$ where $Q \in \mathbb{N}$ and $p$ is an \textbf{integer exponential-polynomial term}, which takes the form
\begin{equation} \label{eq:exp-poly}
\lambda_1^kk^{d_1}\transpose{\vec{a}_1}\vec{x} + \dotsi + \lambda_m^kk^{d_m}\transpose{\vec{a}_m}\vec{x}
\end{equation}
where $\lambda_i \in \spec(A)$, $d_i \in \mathbb{N}$, and $\vec{a}_i \in \mathbb{Z}^n$.\footnote{Technically, we have
$\frac{1}{Q}(\lambda_1^kk^{d_1}\transpose{\vec{a}_1} + \dotsi + \lambda_m^kk^{d_m}\transpose{\vec{a}_m}) = \transpose{\vec{c}}A^{k}\vec{x}$ for all $k$ greater than rank of the highest-rank generalized eigenvector of 0, but since we are only interested in the asymptotic behavior of $A$ we can disregard the first steps of the computation.}

\subsubsection{Characteristic sequences for inequalities}
 \label{sec:dominant-term-analysis}

 Our method for computing characteristic sequences for inequalities is a variation of Tiwari's method for deciding termination of linear loops with real eigenvalues \cite{CAV:Tiwari2004}.
 
 First, suppose that $\vec{p}(\vec{x},k)$ is an integer exponential-polynomial of the form in Eq.~(\ref{eq:exp-poly}) such that each $\lambda_i$ is a \textit{positive} integer. Further suppose that the summands are
ordered by asymptotic growth, with the dominant term appearing
earliest in the list; i.e., for $i < j$ we have either $\lambda_i >
\lambda_j$, or $\lambda_i = \lambda_j$ and $d_i > d_j$.  If we imagine that
the variables $\vec{x}$ are fixed to some $\vec{x}_0 \in \mathbb{Z}^n$,
then we see that 
$p(\vec{x}_0,k)$ is either identically zero or has finitely many zeros, and therefore its sign is eventually stable.  Furthermore, the sign of $p(\vec{x}_0,k)$ as $k$ tends to $\infty$ is simply the sign of its \textit{dominant term} -- that is, the sign of $\transpose{\vec{a}_i}\vec{x}_0$ for the least $i$ such that $\transpose{\vec{a}_i}\vec{x}_0$ is non-zero.  Thus, we may define a function
 $\DTA$ that maps any exponential-polynomial term $p(\vec{x},k)$ (with positive integral $\lambda_i$) to an LIA formula such that for any $\vec{x}_0 \in \mathbb{Z}^n$,
$\vec{x}_0 \models \DTA(p)$ holds if and only if
$\vec{p}(\vec{x}_0,k)$ is eventually non-negative ($\vec{p}(\vec{x}_0,k) \geq 0$ for all but finitely many $k \in \mathbb{N}$). $\DTA$ is defined as follows:
\begin{align*}
  \DTA(0) & \defeq \true\\
  \DTA(\lambda^kk^d\transpose{\vec{a}}\vec{x} + p) &\defeq \transpose{\vec{a}}\vec{x} \geq 1 \lor \left(\transpose{\vec{a}}\vec{x} = 0 \land \DTA(p)\right)
\end{align*}

Finally, we define the characteristic sequence of an inequality atom as follows.
An inequality $t_1 \leq t_2$ over the variables $\vec{x}$ can be written as $\transpose{\vec{c}}\vec{x} + d \geq 0$
   for $\vec{c} \in \mathbb{Z}^n$ and $d \in \mathbb{Z}$.  Let 
   $\frac{1}{Q_{\textit{even}}}p_{\textit{even}}(\vec{x},k)$ and $\frac{1}{Q_{\textit{odd}}}p_{\textit{odd}}(\vec{x},k)$ be the closed forms of
   $\transpose{\vec{c}}A^{2k}(\vec{x})$
   and
   $\transpose{\vec{c}}A^{2k + 1}(\vec{x})$, respectively; by 
   splitting into ``even'' and ``odd'' cases, we ensure that the exponential-polynomial terms 
   $p_{\textit{even}}(\vec{x},k)$ and $p_{\textit{odd}}(\vec{x},k)$ have only \textit{positive} $\lambda_i$
   and thus are amenable to the dominant term analysis $\DTA$ described above.
   Then we define:
\begin{equation*}
  \chi \left( \transpose{\vec{c}}\vec{x} + d \geq 0, A \right) \defeq \left( \DTA(p_{\textit{even}}(\vec{x},k) + dQ_{\textit{even}}), 
  \DTA( p_{\textit{odd}}(\vec{x},k)+ dQ_{\textit{odd}}) \right)^\omega
\end{equation*}

\begin{mexample}
Consider the matrix $A$ and its exponential $A^k$ below:
\[
 A\left(\begin{bmatrix}
 x' \\ y' \\ z' \\ a' \\ b'
 \end{bmatrix}\right)=
 \begin{bmatrix}
 1 & 1 & 0 & 0 & 0\\
 0 & 1 & 1 & 0 & 0\\
 0 & 0 & 1 & 0 & 0\\
 0 & 0 & 0 & -3 & 0 \\
 0 & 0 & 0 & 0 & 2 \\
 \end{bmatrix}
 \begin{bmatrix}
 x \\ y \\ z \\ a \\ b
 \end{bmatrix}
 \]
 \[
  A^k\left(\begin{bmatrix}x\\y\\z\\a\\b\end{bmatrix}\right) =
 \begin{bmatrix}
 1 & k & \frac{k(k-1)}{2} & 0 & 0 \\
 0 & 1 & k & 0 & 0\\
 0 & 0 & 1 & 0 & 0\\
 0 & 0 & 0 & (-3)^k & 0 \\
 0 & 0 & 0 & 0 & 2^k
 \end{bmatrix}
 \begin{bmatrix}
 x \\ y \\ z \\ a \\ b
 \end{bmatrix}
 =
 \begin{bmatrix}
 \frac{1}{2}(zk^2 + (2y - z)k + 2x)\\
 zk + y\\
 z \\
 (-3)^k a\\
 2^k b
 \end{bmatrix}
 \]
First we compute the characteristic sequence $\chi(x \geq 0,A)$.  Applying  the dominant term analysis of the closed form of $x$ yields
\[
  \DTA\left( z k^2 + \left( 2y - z\right) k + x \right) = \left(\begin{array}{ll}&  z > 0  \\ \lor &(z = 0 \land 2y - z > 0) \\ \lor &(z = 0 \land 2y - z = 0 \land x \geq 0)\end{array} \right) \ ,
  \]
  Since the closed form involves only positive exponential terms, we need not split into an even and odd case, and we simply have:
  \[
  \chi(x \geq 0, A) =  (z > 0  \lor (z = 0 \land 2y - z > 0)  \lor (z = 0 \land 2y - z = 0 \land x \geq 0))^\omega \]

Next we compute the characteristic sequence $\chi(a-b \geq 0, A)$, which does require a case split.  Applying dominant term analysis of the closed form of $(a-b)$ yields
\begin{align*}
  \DTA( a \cdot (-3)^{2k} - b \cdot 2^{2k}) &= a > 0 \lor (a=0 \land -b \geq 0) \\
  \DTA( a \cdot (-3)^{2k+1} - b \cdot 2^{2k+1}) &= -a > 0 \lor (-a=0 \land -b \geq 0) \ .
\end{align*}
and thus we have
\[
    \chi(a - b \geq 0, A) = (a > 0 \lor (a=0 \land -b \geq 0), -a > 0 \lor (-a=0 \land -b \geq 0))^\omega \ . \qedhere
\]
\end{mexample}

\subsubsection{Characteristic sequences for divisibility atoms} \label{sec:dta-for-divisibility-atoms}

Last we show how to define $\chi$ for divisibility atoms $n \mid t$.  Write the term $t$ as $\transpose{\vec{c}}\vec{x}+d$ and let the closed form of $\transpose{\vec{c}}A^k(\vec{x})$ be
\[ \frac{1}{Q}(\lambda_1^kk^{d_1}\transpose{\vec{a}_1}\vec{x} + \dotsi + \lambda_m^kk^{d_m}\transpose{\vec{a}_m}\vec{x})\ .\]
The formula $n \mid \transpose{\vec{c}}A^k(\vec{x})+d$ is equivalent to
$Qn \mid \lambda_1^kk^{d_1} \transpose{\vec{a}_1}\vec{x} + \dotsi + \lambda_m^kk^{d_m}\transpose{\vec{a}_m}\vec{x} + Qd$.  For any $i$, the sequence
$\langle\lambda_i^kk^{d_i} \bmod Qn\rangle_{k=0}^\infty$ is ultimately
periodic, since (1) $\langle k \bmod Qn \rangle_{k=0}^\infty = (0, 1, \dots, Qn -1)^\omega$, (2)
$\langle \lambda_i^k \mod Qn \rangle_{k=0}^\infty$ is ultimately periodic (with period and transient length bounded above by $Qn$)\footnote{An infinite sequence $s_0, s_1, s_2, \ldots$ is \textit{ultimately periodic}, if there exists $N$ such that $s_{N}, s_{N+1}, s_{N+2}, \ldots $ is a periodic sequence. We call $N$ the transient length of this sequence. }, and (3) ultimately periodic sequences are closed under pointwise product.  It follows that for each $i$, there is a periodic sequence of integers $\tuple{z_{i,k}}_{k=0}^\infty$ that agrees with $\langle\lambda_i^kk^{d_i} \bmod Qn\rangle_{k=0}^\infty$ on all but finitely many terms.  Finally, we take
\[ \chi(n \mid t,A) \defeq \langle Qn \mid z_{1,k}\transpose{\vec{a}_1}\vec{x} + \dots + z_{m,k}\transpose{\vec{a}_m}\vec{x} + Qd \rangle_{k=0}^\infty\ . \]

\begin{mexample}
Consider matrix $A$ and the closed form of its exponents below
\[
 A \left(\begin{bmatrix}
 x \\ y \\ z
 \end{bmatrix}\right)=
 \begin{bmatrix}
 1 & 1 & 0 \\
 0 & 1 & 0 \\
 0 & 0 & 5 \\
 \end{bmatrix}
 \begin{bmatrix}
 x \\ y \\ z
 \end{bmatrix}
\hspace*{1cm}
A^k\left(
 \begin{bmatrix}
 x \\ y \\ z
 \end{bmatrix} \right) =
 \begin{bmatrix}
 1 & k & 0 \\
 0 & 1 & 0 \\
 0 & 0 & 5^k \\
 \end{bmatrix}
 \begin{bmatrix}
 x \\ y \\ z
 \end{bmatrix}
\]
We show the characteristic sequences for some divisibility atoms w.r.t $A$:
\begin{align*}
    \chi(3 \mid x, A) &= ( 3 \mid x, 3 \mid x+y, 3 \mid x + 2 y)^\omega \\
     \chi(3 \mid x+2, A) &= ( 3 \mid x+2, 3 \mid x+y+2, 3 \mid x + 2 y+2)^\omega \\
      \chi(3 \mid z, A) &= ( 3 \mid z, 3 \mid 2 z)^\omega  \hspace*{3.5cm} \qedhere
\end{align*}
\end{mexample}

%% file: cta.tex
\section{A conditional termination analysis for programs} \label{sec:summary}
This section demonstrates how the results from Sections~\ref{sec:linear-abs}
and~\ref{sec:dta} can be combined to yield a conditional
termination analysis that applies to general programs.

\subsubsection{Integer-spectrum restriction for $\QDLTS{}$} \label{sec:z-restriction-of-qdlts}
Section~\ref{sec:linear-abs} gives a way to compute a $\QDATS{}$-reflection of any transition formula.
Yet the analysis we developed in Section~\ref{sec:dta} 
only applies to linear dynamical systems with integer spectrum.
We now show how to bridge the gap.
Let $V$ be a $\QDATS{}$.  As discussed in Section~\ref{sec:qdats}, we may homogenize $V$ to obtain a $\QDLTS$ $T$.
Define $\Int(T)$ to be the space spanned by the generalized (right) eigenvectors of $T|_\omega$ that correspond to integer eigenvalues:
\[
    \Int(T) \defeq \textit{span}(\{ x \in \dom^\omega(T) : \exists r \in 
    \mathbb{N}^+, \lambda \in \mathbb{Z}. (T|_\omega - \underline{\lambda})^r(x) = 0  \})
\]

Since $\Int(T)$ is invariant under $T|_\omega$ and thus $T$, $T$ defines a linear map $\IntR{T} : \Int(T) \rightarrow \Int(T)$, and by construction $\IntR{T}$ has integer spectrum.  The following lemma justifies the restriction of our attention to the subspace $\Int(T)$.

\begin{theoremEnd}{lemma}
\label{lem:integer-spec-leads-to-soundess}
 Let $F$ be a transition formula, let $\tuple{V, s}$ be a $\QDATS{}$-reflection of $F$, and 
  let $\tuple{T,h} = \Homogenize(V)$.
 For any state $v \in \dom^\omega(F)$, we have $h(s(v)) \in \Int(T)$.
\end{theoremEnd}
\begin{mproofEnd}
Define $t \defeq h \circ s$; observe that $t$ is a simulation from $F$ to $T$.

By the spectral theorem, there is a basis $\vec{b}_1,\dots,\vec{b}_n$ for $\dom^\omega(T)$ such that for each $i$ there is some $\lambda_i \in \mathbb{Q}$ and some $r_i \in \mathbb{N}^{\geq 1}$ such that 
$(T|_\omega - \underline{\lambda_i})^{r_i}(\vec{b}_i) = 0$.
We now argue that without loss of generality, we can further suppose that for all $v \in \dom^\omega(F)$ such that $t(v) \in \dom^\omega(T)$, there exist \textit{integers} $a_1,\dots,a_n$ such that $t(v) = a_1\vec{b}_1 + \dots + a_n\vec{b}_n$.  First, extend the basis $\vec{b}_1,\dots,\vec{b}_n$ to a basis $\vec{b}_1,\dots,\vec{b}_m$ for the whole space $S_T$.  Let $\set{v_x}_{x \in X}$ denote the standard basis for $S_F = X \rightarrow \mathbb{Q}$, where where $v_x$ maps $x$ to 1 and all variables other than $x$ to $0$.  Since $\vec{b}_1,\dots,\vec{b}_m$ are rational vectors, there exist rationals $z_1,\dots,z_m$ such that for each $v_j$, $t(v_j)$ is an integral linear combination of $z_1\vec{b}_1,\dots,z_m\vec{b}_m$; $z_1\vec{b}_1,\dots,z_m\vec{b}_n$ is a basis for $\dom^\omega(T)$ with the desired property.  To verify that for each $v \in \dom^\omega(F)$, $t(v)$ is an integral combination of $z_1\vec{b}_1,\dots,z_m\vec{b}_m$, observe that
since $v \in \dom^\omega(F)$, $v$ is an integral valuation and so $v$ is an integral combination of $\set{v_x}_{x \in X}$, and therefore $t(v)$ is an integral combination of $\vec{b}_1,\dots,\vec{b}_m$.  Since $t(v) \in \dom^\omega(T)$, the coefficients of $\vec{b}_{n+1},\dots,\vec{b}_m$ are zero.

    Since for each $i$ we have
    $(T|_\omega - \underline{\lambda_i})^{r_i}(\vec{b}_i) = 0$, we have for any $k$ and any $i$ such that $\lambda_i \neq 0$,
    $T|_\omega^k(\vec{b}_i) = \lambda_i^k\vec{b}_i + \vec{v}$, where $\vec{v}$ satisfies 
    $(T|_\omega - \underline{\lambda_i})^{r_i-1}(\vec{v}) = 0$ (i.e., $\vec{v}$ is a generalized eigenvector of lower rank).  This can be shown by induction on $k$.  The base case $k=0$ is trivial.  For the inductive step, suppose
    $T|_\omega^k(\vec{b}_i) = \lambda_i^k\vec{b}_i + \vec{v}$
    with $(T|_\omega - \underline{\lambda_i})^{r_i-1}(\vec{v}) = 0$.
    Then
    \begin{align*}
        T|_\omega^{k+1}(\vec{b}_i) & = T|_\omega(\lambda_i^k\vec{b}_i + \vec{v})\\
        & = \lambda_i^k(T|_\omega(\vec{b}_i)) + T|_\omega(\vec{v})\\
        &= \lambda_i^{k+1}\vec{b}_i + (T|_\omega(\vec{v}) + \lambda_i^{k}(T|_\omega(\vec{b}_i)) -  \lambda_i^{k+1}\vec{b}_i)\\
        &=  \lambda_i^{k+1}\vec{b}_i + (T|_\omega(\vec{v}) + \lambda_i^{k}(T|_\omega - \underline{\lambda_i})(\vec{b}_i))
    \end{align*}
    and we may verify
    \[
    \begin{array}{ll}
  &    (T|_\omega - \underline{\lambda_i})^{r_i-1}\left(T|_\omega(\vec{v}) + \lambda_i^{k}(T|_\omega - \underline{\lambda_i})(\vec{b}_i)\right)\\
  =&T|_\omega((T|_\omega - \underline{\lambda_i})^{r_i-1}(\vec{v})) + (T|_\omega - \underline{\lambda_i})^{r_i}(\vec{b}_i)\\
  =& 0.
    \end{array}
    \]
    
Suppose $v_0 \in \dom^\omega(F)$.  Then there exists an infinite trajectory
  \[ v_0 \rightarrow_F v_1 \rightarrow_F v_2 \rightarrow_F \dots \]
   and so there is an infinite trajectory
   \[ t(v_0) \rightarrow_T t(v_1) \rightarrow_T t(v_2) \rightarrow_T \dots \]
   and so $t(v_0) \in \dom^\omega(T)$.
  We want to prove $t(v_0) \in \Int(T)$.

  Since for all $k$ we have $t(v_k) \in \dom^\omega(T)$, there exists integers $c_{k,1},...,c_{k,n}$ such that $t(v_k) = c_{k,1}\vec{b}_1 + \dots + c_{k,n}\vec{b}_n$.
  Since $T|_\omega$ is the restriction of $T$ to $\dom^\omega(T)$ and
  $T|_\omega$ is linear, we have
  \[  t(v_k) = (T|_\omega)^k(t(v)) = c_{0,1}(T|_\omega)^k(\vec{b}_1) + \dots + c_{0,n}(T|_\omega)^k(\vec{b}_n) \]

  Suppose for a contradiction that there is some $i$ with $c_{0,i} \neq 0$ and $\lambda_i$ is not an integer.  Without loss of generality, further suppose that $\vec{b}_i$ has the greatest rank $r_i$ among all such indices.  We will show that $c_{k,i} = \lambda_i^kc_{0,i}$ for all $k$, which is a contradiction since each $c_{k,i}$ is an integer and $\lambda_i$ is not an integer.
  
  By the above, for each $\vec{b}_j$ we have  
  $T|_\omega^k(\vec{b}_j) = \lambda_j^k\vec{b}_j + \vec{v}_j$ for some $\vec{v}_j$ satisfying
    $(T|_\omega - \underline{\lambda_j})^{r_i-1}(\vec{v}_j) = 0$, and we have
      \[  t(v_k) = (T|_\omega)^k(t(v)) = c_{0,1}(\lambda_1^k\vec{b}_1 + \vec{v}_1) + \dots + c_{0,n}(\lambda_n^k\vec{b}_n + \vec{v}_n) \]
    Since $r_i$ is maximal, every $\vec{v}_j$ is orthogonal to $\vec{b}_i$ (w.r.t. the basis $\vec{b}_1,\dots,\vec{b}_n$), and so $c_{k,i}$ is simply $\lambda_i^kc_{0,i}$.

\end{mproofEnd}

\begin{figure}[t]
\centering
  \begin{tikzpicture}
    \node (program) {\begin{minipage}{4.5cm}
\begin{lstlisting}[style=base,numbers=left,xleftmargin=2em]
c := 0    
while (x % 2 == 0) do
 x = x / 2
 c = c + 1
\end{lstlisting}
  \end{minipage}};

  \node [right of=program, node distance=6cm] (tf) {
  \begin{minipage}{4cm}
  \begin{align*}
        F(x,c,x',c') = &\quad (2 \mid x) \\ 
        &\land (x - 1 \leq 2x' \land 2x' \leq x) \\
        &\land (c' = c + 1)
  \end{align*}
  \end{minipage}
  };
 
  \end{tikzpicture}
  \caption{An example program that contains integer division. \label{fig:cta-eg2}}
\end{figure}

\begin{mexample}
 Figure~\ref{fig:cta-eg2} shows a loop that computes the number of trailing $0$'s in the binary representation of integer $x$ and its corresponding transition formula.
 The homogenization of the $\QDATS{}$-reflection of $F$ is the $\QDLTS{}$ $T$, where:
 \[
 R_T \defeq
  \set{
 \tuple{\begin{bmatrix}x\\c\\h\end{bmatrix}, \begin{bmatrix}x\\c\\h\end{bmatrix}}:
 \begin{bmatrix}
 x' \\ c' \\ h' 
 \end{bmatrix} =
 \begin{bmatrix}
 \frac{1}{2} & 0 & 0 \\
 0 & 1 & 1\\
 0 & 0 & 1
 \end{bmatrix}
 \begin{bmatrix}
 x \\ c \\ h
 \end{bmatrix}}
  \]
  The $\omega$-domain of $T$ is the whole state space $\mathbb{Q}^3$.
Since the eigenvector  $\transpose{\begin{bmatrix}
1 & 0 & 0
\end{bmatrix}}$ of the transition matrix 
corresponds to a non-integer eigenvalue $\frac{1}{2}$,
the $x$-coordinate of states  in $\Int(T)$ must be $0$; i.e.,
$\Int(T) = \{ (x, c, y): x = 0 \}$.
We conclude that $x \neq 0$ is a sufficient condition for the loop to terminate.
\end{mexample}

\subsubsection{The mortal precondition operator}
 Algorithm~\ref{alg:end-to-end-mp} shows how to compute a mortal precondition for an LIA transition formula $F(\vec{x},\vec{x}')$ (i.e., a sufficient condition for which $F$ terminates).  The algorithm operates as follows.  First, we compute a $\QDATS{}$-reflection of $F$, and homogenize to get a $\QDLTS{}$ $T$ and an \textit{affine} simulation $t : F \rightarrow T$.  Let $p$ denote an (arbitrary) projection from $S_T$ onto $\Int(T)$ (so $p$ is a simulation from $T$ to $\IntR{T}$).
 We then compute an LIA formula $G$ which represents the states $\vec{w}$ of $T
 _{\Int(T)}$ such that there is some 
$v \in \dom(F)$ such that
$t(v) \in \Int(T)$ and $p(t(v)) = \vec{w}$.  Letting $(H_0,...,H_P)^\omega$ be the characteristic sequence $\chi(G, \IntR{T})$, we have that for any $v \in \dom^\omega(F)$,  $t(v)$ must belong to $\Int(T)$ and $p(t(v))$ satisfies each $H_i$, so we define
\[ \mp(F) \defeq \{ v \in S_F : t(v) \notin \Int(T) \text{ or } v \not\models \bigwedge_i H_i(p(t(\vec{x}))) .\]

Within the context of the algorithm, we suppose that states of $F$ are represented by $n$-dimensional vectors, states of $T$ are represented as $m$-dimensional vectors, and state of $\IntR{T}$ are represented as $q$-dimensional vectors.  The affine simulation $t$ is represented in the form $\vec{x} \mapsto A\vec{x} + \vec{b}$, where $A \in \mathbb{Z}^{m \times n}$ and $\vec{b} \in \mathbb{Z}^m$, the projection $p$ as a $\mathbb{Z}^{q \times m}$ matrix, and the linear map $\IntR{T}$ as a $\mathbb{Q}^{q \times q}$ matrix.  The fact that $p$ and $t$ have all integer (rather than rational) entries is without loss of generality, since any simulation can be scaled by the least common denominator of its entries.

\begin{algorithm}
  \Input{A transition formula $F(\vec{x},\vec{x'}) \in \TF$ in linear integer arithmetic.}
  \Output{A mortal precondition $\mp(F)$ for $F$.}
  $A \gets \aff(F)$ \tcc*{Affine hull \cite{VMCAI:RSY2004}; Lemma~\ref{lem:best-ats}}
  $\tuple{D,d} \gets \alpha_{\Det(A)}(A)$ \tcc*{Determinize; Lemma~\ref{lem:determinization}}
  $\tuple{V,q} \gets $ $\QDATS{}$-reflection of $D$ \tcc*{Algorithm~\ref{alg:rational-spec-abstraction}}
  $v \gets q \circ d$ \tcc*{$\tuple{V,v}$ is a $\QDATS{}$-reflection of $F$}
  $\tuple{T,h} \gets \Homogenize(V)$ \tcc*{Homogenization of $V$}
  $t \gets h \circ v$ \tcc*{$t$ is an affine simulation $F \rightarrow T$}
  $p \gets $ (any) linear projection of $S_T$ onto $\Int(T)$\;
  $C \gets $ matrix such that
  $C\vec{w} = 0 \iff \vec{w} \in \Int(T)$\;
  Let $G(\vec{w}) \gets \exists \vec{x}, \vec{x'}.  F(\vec{x}, \vec{x}') \land \vec{w} = p (t( \vec{x})) \land C t( \vec{x}) = 0$\;
  $(H_0(\vec{w}), \dots, H_P(\vec{w}))^\omega \gets \chi(G(\vec{w}), \IntR{T})$ \tcc*{Section~\ref{sec:dta}}
  \Return{$\lnot \left( \left(\bigwedge_i H_i(p(t( \vec{x})))\right) \land C t( \vec{x}) = 0 \right)$}
  \caption{Procedure for computing $\mp(F)$.}
  \label{alg:end-to-end-mp}
  \end{algorithm}

\begin{theoremEnd}[normal]{theorem}[Soundness]
  For any transition formula $F$, for any state $s$ such that $s
  \in \mp(F)$, we have $s \notin \dom^\omega(F)$.
  \label{thm:mp-dta-soundness}
\end{theoremEnd}
\begin{mproofEnd}
  Let $T$, $t$, $p$, $C$, $G$, and $H_0,\dots,H_P$ be as in Algorithm~\ref{alg:end-to-end-mp}.
  We prove the contrapositive: we assume $v \in \dom^\omega(F)$ and prove $v \notin mp(F)$, or equivalently $v \models 
  H_i(p (t( \vec{x})))$ for each $i$ and $t(v) \in \Int(T)$.  We have $t(v) \in \Int(T)$ by Lemma~\ref{lem:integer-spec-leads-to-soundess}, so it remains only to show that $v \models H_i(p (t( \vec{x})))$ for each $i$.
  
  Since $v \in \dom^\omega(F)$, there exists an infinite trajectory of $F$
  starting from $v$: $v \rightarrow_F v_1 \rightarrow_F  v_2 \rightarrow_F \ldots$.
  For any $j$, let $\vec{w}_j = \IntR{T}^j(p(t(v)))$.  Since $p \circ t$ is an (affine) simulation, we have $\vec{w}_j = p(t(v_j))$ for all $j$.  It follows that for any $j$, we have
  $[v_j,v_{j+1}] \models F(\vec{x},\vec{x}') \land \vec{w}_j = p(t(\vec{x}_j)) \land Ct(\vec{x}_j) = 0$, and so
  $G(\vec{w}_j) = \exists \vec{x},\vec{x}'. F(\vec{x},\vec{x}') \land \vec{w}_j = p(t(\vec{x})) \land Ct(\vec{x}) = 0$ holds for all $j$.  By Theorem~\ref{thm:periodic-truth-value-of-guard-under-linear-map},
  $H_i(p(t( \vec{x})))$ holds for all $H_i$.
\end{mproofEnd}


Monotonicity is a desirable property for termination analysis:
it guarantees that more information into the analysis always leads to better results.
In our context, monotonicity also guarantees that if we cannot prove termination 
using the $\mp$ operator that we defined, then \textit{any} linear abstraction of the loop
does not terminate.
\begin{theoremEnd}[category=monotonicity]{theorem}[Monotonicity]
  For any transition formulas $F_1$ and $F_2$ such that $F_1 \models F_2$, we
  have $\mp(F_2) \models \mp(F_1)$.
  \label{thm:mp-dta-monotonicity}
\end{theoremEnd}
\begin{mproofEnd}
Let $\tuple{A_1,m_1}$ and $\tuple{A_2,m_2}$ be $\QDATS{}$-reflections of $F_1$ and $F_2$, respectively.  Let $\tuple{T,h_1} = \Homogenize(A_1)$ and $\tuple{U,h_2} = \Homogenize(A_2)$ be their homogenizations.  Define $t \defeq h_1 \circ m_1$ and $u \defeq h_2 \circ m_2$ to be affine simulations $t : F_1 \rightarrow T$ and $u : F_2 \rightarrow U$.
  Since $F_1 \models F_2$ and $ \tuple{A_2, m_2}$ is a $\QDATS{}$-abstraction of $F_2$, it is also an abstraction of $F_1$.
  Since $ \tuple{A_1, m_1}$ is the $\QDATS{}$-reflection of $F_1$, there exists a
  simulation $\bar{m} : A_1 \rightarrow A_2$ 
  such that $m_2 = \bar{m} \circ m_1$.  Then $\overline{u} \defeq \Hom(\overline{m})$ is a simulation from $T$ to $U$ such that
  $u = \bar{u} \circ t$.

  We prove the contrapositive 
  $\neg \mp(F_1) \models \neg \mp(F_2)$,
  which by Lemma~\ref{lem:semantics-of-char-sequences} is equivalent\footnote{The definition of $\mp$ is given in Algorithm~\ref{alg:end-to-end-mp}. For the case of $F_1$, we use the state formula $\exists \vec{x}, \vec{x}'. F_1(\vec{x}, \vec{x}') \land \vec{y} = t(\vec{x})$ as the LIA formula $G$ in Lemma~\ref{lem:semantics-of-char-sequences}. Similarly, we use $\exists \vec{x}, \vec{x}'. F_2(\vec{x}, \vec{x}') \land \vec{w} = u(\vec{x})$ as the LIA formula for $F_2$.} to
  showing that for all $v \in S_{F_1}$ such that $t(v) \in \mathrm{SAT}^\infty(T, \exists \vec{x}, \vec{x}'. F_1(\vec{x}, \vec{x}') \land \vec{y} = t(\vec{x}))$, we have $u(v) \in \mathrm{SAT}^\infty(U, \exists \vec{x}, \vec{x}'. F_2(\vec{x}, \vec{x}') \land \vec{w} = u(\vec{x}))$

  Let $v \in S_{F_1}$ such that  $t(v) \in \mathrm{SAT}^\infty(T, \exists \vec{x}, \vec{x}'. F_1(\vec{x}, \vec{x}') \land \vec{y} = t(\vec{x}))$.
  By the definition of $\mathrm{SAT}^\infty$, we know that there exists 
  an infinite trajectory $\vec{y}_1, \vec{y}_2, \dots \in \Int(T)$ such that 
  $t(v) \rightarrow_T \vec{y}_1 \rightarrow_T \vec{y}_2 \rightarrow_T \dots$ 
  and a $K \in \mathbb{N}$ such that for all $k \geq K$,
  $\exists \vec{x}, \vec{x}'. F_1(\vec{x}, \vec{x}') \land \vec{y}_k = t(\vec{x})$ holds.
  Since $\bar{u} : T \rightarrow U$ is a simulation, there is a corresponding trajectory in $U$:
  $\bar{u}(t(\vec{x})) \rightarrow_U \vec{w}_1 = \bar{u}(\vec{y}_1) \rightarrow_U \vec{w}_2 = \bar{u}(\vec{y}_2) \rightarrow_U \dots$.

  Let $k \geq K$.  Since 
  $\exists \vec{x}, \vec{x}'. F_1(\vec{x}, \vec{x}') \land \vec{y}_k = t(\vec{x})$ holds and $F_1 \models F_2$, we must have
  $\exists \vec{x}, \vec{x}'. F_2(\vec{x}, \vec{x}') \land \vec{y}_k = t(\vec{x})$ and thus
  $\exists \vec{x}, \vec{x}'. F_2(\vec{x}, \vec{x}') \land \overline{u}(\vec{y}_k) = \overline{u}(t(\vec{x}))$.
  Since $\overline{u}(\vec{y}_k) = \vec{w}_k$ and
  $\overline{u}(t(\vec{x})) = u(\vec{x})$, and the above holds for all $k \geq K$, it follows that
  $u(\vec{x}) \in \mathrm{SAT}^\infty(U, \exists \vec{x}, \vec{x}'. F_2(\vec{x}, \vec{x}') \land \vec{w} = u(\vec{x}))$.
\end{mproofEnd}

%% file: evaluation.tex
\section{Evaluation} \label{sec:evaluation}

Section~\ref{sec:summary} shows how to compute mortal preconditions for transition formulas.  Using the framework of algebraic termination analysis \cite{arxiv}, we can ``lift'' the analysis to compute mortal preconditions for whole programs.  The essential idea is to compute summaries for loops and procedures in ``bottom-up'' fashion, apply the mortal precondition operator from Section~\ref{sec:summary} to each loop body summary, and then propagate the mortal preconditions for the loops back to the entry of the program (see \cite{arxiv} for more details).  We can verify that a program terminates by using an SMT solver to check that its mortal precondition is valid.


We have implemented our conditional termination analysis (BLACT, \textbf{B}est \textbf{L}inear \textbf{A}bstraction
for \textbf{C}onditional \textbf{T}ermination).
We compare the performance of our analysis
against 2LS~\cite{TOPLAS:CDKSW2018}, Ultimate Automizer~\cite{TACAS:DietschHNSS20} and CPAchecker~\cite{CPACheckerTermination}, the top three
competitors in the termination category of \href{https://sv-comp.sosy-lab.org/2020}{Competition on Software Verification (SV-COMP) 2020}. 
We also compare with ComPACT \cite{arxiv}, which uses the same
algebraic termination analysis framework as BLACT, but uses a 
different method to compute conditional termination arguments for individual loops.

Experiments are run on a virtual machine with Ubuntu 18.04, with 
    a single-core Intel Core i5-7267U @ 3.10GHz CPU and 3 GB of RAM.
    All tools were run with a time limit of $10$ minutes. 
    CPAChecker was run with a heap memory limit of 2500 MB.
    
\paragraph{Benchmarks}
We tested on a suite of 263 programs divided into 4 categories.
The \texttt{termination} and \texttt{recursive} suites contain small programs with challenging termination arguments, while the \texttt{polybench} suite contains 
larger real-world programs that have relatively simple termination arguments.
The \texttt{termination} category consists of the
\textit{non-recursive, terminating} benchmarks from SV-COMP 2020 in the \texttt{Termination-MainControlFlow} suite.  
The \texttt{recursive} category consists of the \textit{recursive, terminating} benchmarks from the \texttt{recursive} directory and \texttt{Termination-MainControlFlow}. Note that 2LS does not handle recursive programs, so we exclude it from the \texttt{recursive} category.
Finally, we created a new test suite \texttt{linear} consisting of terminating programs
whose termination can be proved by computing best linear abstractions and analyzing
the asymptotic behavior of these linear abstractions.
This suite includes: all examples of loops with multi-phase ranking functions from \cite{CAV:BG2017}, loops with disjunctive or modular arithmetic guards, loops that model integer division and remainder calculation, etc.

\begin{figure*}[t]
  \footnotesize\noindent
    \begin{tabular}{lc|cr|cr|cr|cr} 
\toprule
            &         & \multicolumn{2}{c|}{BLACT}     & \multicolumn{2}{c|}{2LS}   & \multicolumn{2}{c|}{UAutomizer} & \multicolumn{2}{c}{CPAChecker}  \\
benchmark   & \#tasks & \#correct    & time            & \#correct & time           & \#correct     & time            & \#correct & time                \\ 
\midrule
termination & 171     & 92           & \textbf{78.8}   & 115       & 1966.0         & \textbf{161}  & 4772.2          & 126       & 12108.6             \\
recursive   & 42      & 4            & \textbf{32.8}   & -         & -              & \textbf{30 }  & 1781.7          & 23        & 530.6               \\
polybench   & 30      & \textbf{30}  & \textbf{82.8}   & 0         & 7602.7         & 0             & 16241.6         & 0         & 4035.8              \\
linear      & 20      & \textbf{20}  & 25.0            & 6         & \textbf{17.6}  & 8             & 2841.3          & 3         & 3470.7              \\ 
\midrule
Total       & 263     & 146          & \textbf{219.4}  & 121       & 9586.3         & \textbf{199 } & 25636.8         & 152       & 20145.7             \\
\bottomrule
\end{tabular}
    \caption{Termination verification benchmarks; time in seconds.  
    \label{tab:all-tools-all-benchmarks}}
\end{figure*}
\begin{figure*}[t]
\centering
  \footnotesize\noindent
   \begin{tabular}{lc|cr|cr|cr} 
\toprule
                    &        & \multicolumn{2}{c|}{BLACT} & \multicolumn{2}{c|}{ComPACT} & \multicolumn{2}{c}{BLACT+ComPACT}  \\
                    & \#tasks & \#correct     & time                            & \#correct     & time                              & \#correct      & time         \\ 
\midrule
termination         & 171    & 92           & \textbf{78.8}                   & 141          & 99.4                              & \textbf{145}  & 105.6        \\
recursive           & 42     & 4            & \textbf{32.8}                   & 31           & 104.7                             & \textbf{32}   & 106.2        \\
polybench           & 30     & \textbf{30}  & \textbf{82.8}                   & \textbf{30}  & 177.5                             & \textbf{30}   & 192.7        \\
linear & 20     & \textbf{20}  & \textbf{25.0}                   & 15           & 209.9                             & \textbf{20}   & 78.6         \\ 
\midrule
Total               & 263    & 146          & \textbf{219.4}                  & 217          & 591.5                             & \textbf{227}  & 483.1        \\
\bottomrule
\end{tabular}
    \caption{Comparing BLACT and ComPACT; time in seconds.
    \label{tab:blact-vs-compact}}
\end{figure*}

\paragraph{How does BLACT compare with the state-of-the-art?}
The comparison of BLACT against existing termination analysis tools across all test suites is shown in
Figure~\ref{tab:all-tools-all-benchmarks}. 
BLACT uses substantially less time than 2LS, UAutomizer, and CPAChecker.  BLACT is able to prove fewer benchmarks on the \texttt{termination} and \texttt{recursive} suites--these benchmarks are designed to have difficult termination arguments, which most tools approach by using a portfolio of different termination techniques (e.g., UAutomizer synthesizes
linear, nested, multi-phase, lexicographic and piecewise ranking functions).  We investigate the use of BLACT in a portfolio solver in the following.


BLACT solves all tasks within the \texttt{polybench} suite, which contains 
larger numerical programs that have simple termination arguments.  2LS, UAutomizer, and CPAChecker are unable to solve any, as they run out of time or memory.  Nested loops are a problematic pattern that appears in these tasks, e.g.,
\begin{lstlisting}[style=base]
for(int i = 0; i < 4096; i += step)
  for (int j = 0; j < 4096; j += step)
    // no modifications to i, j, or step
\end{lstlisting}
For such loops, BLACT is guaranteed to synthesize a conditional termination argument that is \textit{at least} as weak as \cinline{step > 0} (regardless of the contents of the inner loop) by monotonicity and the fact that the loop body formula entails $i < 4096 \land i' = i + \cinline{step} \land \cinline{step}' = \cinline{step}$.
UAutomizer, CPAChecker, and 2LS cannot make such theoretical guarantees.

The performance of 2LS, UAutomizer, and CPAChecker on the \texttt{linear} suite demonstrate that BLACT is capable of proving termination of programs that lie outside the boundaries of any other tool.



\paragraph{Can BLACT improve a portfolio solver?}
We compare BLACT and ComPACT in Figure~\ref{tab:blact-vs-compact}.
ComPACT implements several termination techniques, and BLACT can be added as a plug-in.
Adding BLACT to the existing portfolio of methods implemented in ComPACT, 
we can solve 10 additional tasks while adding negligible runtime overhead.  In fact, adding BLACT to the portfolio \textit{decreases} the amount of time it takes for ComPACT to complete the \texttt{linear} suite.  
Note that the combined tool also yields the 
most powerful tool among all those we tested, 
for all except the \texttt{termination} category.


%% file: related.tex
\section{Related work} \label{sec:related}

\paragraph{Termination analysis of linear loops}

The universal termination problem for linear loops (or \textit{total
 deterministic affine transition systems}, in the terminology of
Section~\ref{sec:dta}) was posed by Tiwari \citet{CAV:Tiwari2004}.  The case
of linear loops over the reals was resolved by Tiwari \citet{CAV:Tiwari2004},
over the rationals by
Braverman \citet{CAV:Braverman2006}, and finally over the
integers by
Hosseini et al. \citet{ICALP:HOW2019}.
In principle, we can combine any of these techniques with our algorithm for computing
$\DATS{}$-reflections of transition formulas to yield a sound (but incomplete) termination analysis.  The significance of computing a $\DATS{}$-reflection (rather than just ``some'' abstraction) is that is provides an algorithmic completeness result: if it is possible to prove termination of a loop by exhibiting a terminating linear dynamical system that simulates it, the algorithm will prove termination.

The method introduced in Section~\ref{sec:dominant-term-analysis} to compute characteristic sequences of inequalities is based on the method that Tiwari used to prove decidability of the universal termination problem for linear
loops with (positive) real spectra \citet{CAV:Tiwari2004}.  Tiwari's condition of having
\textit{real} spectra is strictly more general than the
\textit{integer} spectra used by our procedure; requiring that the
spectrum be integer allows us express the $\mathbf{DTA}$ procedure in linear \textit{integer} arithmetic rather than real arithmetic.  Similar procedures appear also in \cite{POPL:KBCR2019,CAV:FG2019}.  We note in particular that our results in
Sections~\ref{sec:dta} and~\ref{sec:summary}
subsume 
Frohn and Giesl's decision procedure for universal
termination for upper-triangular linear loops \citet{CAV:FG2019}; since every
rational upper-triangular linear loop has a rational spectrum (and is
therefore a $\QDATS{}$), the mortal precondition computed for any rational upper-triangular linear loop is valid iff the loop is universally terminating.

\paragraph{Linear abstractions}
The formulation of ``best abstractions'' using reflective subcategories is based on the framework developed in \cite{SAS:Kincaid2018}.  A variation of this method was used by Kincaid and Silverman in the context of invariant generation by finding (weak) reflections of linear rational arithmetic formulas in the category of rational vector addition systems \cite{CAV:SK2019}.  This paper is the first to apply the idea to termination analysis.

Kincaid et al. give a method for extracting polynomial recurrence (in)equations that are entailed by a transition formula \cite{PACMPL:KCBR2018}.  The algorithm can also be applied to compute a $\mathbf{TDATS}$-abstraction of a transition formula.  The procedure does not guarantee that the $\mathbf{TDATS}$-abstraction is a reflection (\textit{best} abstraction); Proposition~\ref{prop:no-best-lts} demonstrates that no such procedure exists.  In this paper, we generalize the model to allow non-total transition systems, and show that best abstractions do exist.  The techniques from Section~\ref{sec:linear-abs} can be used for invariant generation, improving upon the methods of \cite{PACMPL:KCBR2018}.

Kincaid et al. show that the category of linear dynamical systems with \textit{periodic rational} spectrum is a reflective subcategory of the category of linear dynamical systems \cite{POPL:KBCR2019}.  A complex number $n$ is periodic rational if $n^p$ is rational for some $p \in \mathbb{N}^+$.  Combining this result with the technique from Section~\ref{sec:linear-abs} yields the result that the category of $\mathbf{DATS}$ with periodic rational spectrum is a reflective subcategory of $\mathbf{TF}$.  The decision procedure from Section~\ref{sec:dta} extends easily to the periodic rational case, which results in a strictly more powerful decision procedure.



\paragraph{Termination analysis}
Termination analysis, and in particular conditional termination analysis, has been widely studied.  Work on the subject can be divided into practical termination analyses that work on real programs (but offer few theoretical guarantees) \cite{CAV:CGLRS2008,POPL:CC2012,SAS:Urban2013,ESOP:UM2014,SAS:UM2014,PLDI:LQC2015,CAV:DU2015,CAV:GG2013,TACAS:BBLORR2017}, and work on simplified model (such as linear, octagonal, and polyhedral loops) with strong guarantees (but cannot be applied directly to real programs) \cite{VMCAI:PR2004,CAV:BMS2005,TACAS:LH2014,CAV:BG2017,CAV:Tiwari2004,ICALP:HOW2019,CAV:Braverman2006}.  This paper aims to help bridge the gap between the two, by showing how to apply analyses for linear loops to general programs, while preserving some of their good theoretical properties, in particular monotonicity.

%% file: appendix.tex
\eject

\appendix

\section{Proofs}

\printProofs

\subsection{Homogenization}
Section~\ref{sec:qdats} gives a simple construction of the homogenization of a $\DATS$ over a state space of the form $\mathbb{Q}^n$.  We now give a general definition (for arbitrary linear spaces) that is more convenient for our technical development.
Suppose that $T$ is a non-empty affine transition system, and let
$\tr{T}{x_0}{x_0'} \in R_T$ be an arbitrary transition (the choice is
irrelevant).  Define $\Hom(T)$ to be the transition system where
\begin{align*}
  S_{\Hom(T)} & \defeq S_T \times \mathbb{Q}\\ R_{\Hom(T)}
  & \defeq \set{ \tuple{\tuple{x,c},\tuple{x',c}}
  : \tr{T}{(x+(1-c)x_0)}{(x'+(1-c)x_0')} }
\end{align*}
Observe that $\Hom(A)$ is linear since $R_{\Hom(A)}$ is an affine
space and contains the origin, and that we have
$\tr{\Hom(A)}{\tuple{x,1}}{\tuple{x',1}}$ if and only if
$\tr{A}{x}{x'}$.  Homogenization extends to a functor: for any
simulation $\sim{s}{A}{B}$ between non-empty affine transition systems, there is
a corresponding simulation $\sim{\Hom(s)}{\Hom(A)}{\Hom(B)}$ with
$\Hom(s)(x,c) \defeq \tuple{s(x),c}$.

\printProofs[homogenization]

\subsection{Computation of $\QDLTS{}$-reflections}
We give an example for the computation of a  $\QDLTS{}$-reflection of a given $\DLTS{}$.
\begin{mexample} \label{ex:best-qdlts-abstraction}
 Consider the $\DLTS{}$ $T$ with transition relation
 \[
 R_T \defeq \set{\tuple{\begin{bmatrix}w\\x\\y\\z\end{bmatrix},\begin{bmatrix}w'\\x'\\y'\\z'\end{bmatrix}} :
 \begin{bmatrix}
 1 & 0 & 0 & 0\\
 0 & 1 & 0 & 0\\
 0 & 0 & 1 & 0\\
 0 & 0 & 0 & 1\\
 0 & 0 & 0 & 0
 \end{bmatrix}
 \begin{bmatrix}x'\\y'\\z'\end{bmatrix} =
  \begin{bmatrix}
 1 & 1 & 0 & 0 \\
 1 & 1 & 0 & 0 \\
 0 & 0 & 0 & 1 \\
 0 & 0 & -1 & 0\\
 1 & -1 & 0 & 0
 \end{bmatrix}
 \begin{bmatrix}w\\x\\y\\z\end{bmatrix}}
 \]
 We can calculate the $\omega$-domain of $T$
 $\dom^\omega(T) = \set{ \transpose{\begin{bmatrix}w&x&y&z\end{bmatrix}} : w = x}$, which has
 a basis $B = \transpose{\begin{bmatrix}1 & 1 & 0 & 0\end{bmatrix}}, \transpose{\begin{bmatrix}0 & 0 & 1 & 0\end{bmatrix}}, \transpose{\begin{bmatrix}0 & 0 & 0 & 1\end{bmatrix}}$.  With respect to $B$, $T|_\omega$ corresponds to the matrix
 \[
 T|_\omega = \begin{bmatrix}
 2 & 0 & 0\\
 0 & 0 & 1\\
 0 & -1 & 0
 \end{bmatrix}
 \]
 and so we have $\spec(T) = \set{2,i,-i}$.  We may calculate $E_\mathbb{Q}(T)$ by finding (generalized) left eigenvectors with eigenvalue 2, the only rational number in $\spec(T)$:
\begin{align*} E_\mathbb{Q}(T) &=
 \set{ \transpose{\vec{v}} : \transpose{\vec{v}}
 \underbrace{\begin{bmatrix}
 1 & 0 & 0\\
 1 & 0 & 0\\
 0 & 1 & 0\\
 0 & 0 & 1
 \end{bmatrix}}_{B}
 \left(
 \underbrace{
 \begin{bmatrix}
 2 & 0 & 0\\
 0 & 0 & 1\\
 0 & -1 & 0
 \end{bmatrix}}_{T|_\omega} - \underbrace{\begin{bmatrix}
 2 & 0 & 0\\
 0 & 2 & 0\\
 0 & 0 & 2
 \end{bmatrix}}_{2I}\right) = 0}\\
 &=
 \textit{span}(\begin{bmatrix}1 & 1 & 0 & 0\end{bmatrix}, \begin{bmatrix}1 & -1 & 0 & 0\end{bmatrix})
 \end{align*}
Finally, we have $\tuple{Q,q} = \alpha_{E_{\mathbb{Q}}(T)}(T)$, where
 \[R_Q = \set{\tuple{\begin{bmatrix}a\\b\end{bmatrix},\begin{bmatrix}a'\\b'\end{bmatrix}}: \begin{bmatrix}1 & 0\\0 & 1\\0&0\end{bmatrix}\begin{bmatrix}a'\\b'\end{bmatrix}
 =\begin{bmatrix}2 & 0\\0 & 0\\0 & 1\end{bmatrix}\begin{bmatrix}a\\b\end{bmatrix}}
 \hspace*{1cm} q = \begin{bmatrix}1 & 1 & 0 & 0\\1 & -1 & 0 & 0\end{bmatrix} \]
 $Q$ is deterministic and has rational spectrum, so $\tuple{Q,q}$ is a $\QDLTS{}$-reflection of $T$.
\end{mexample}

\subsection{Monotonicity}

We study properties of states 
that belongs to $\mp(F)$, which are later used to prove 
the monotonicity theorem (Theorem~\ref{thm:mp-dta-monotonicity}).
First we define a set of initial states of a $\QDLTS$ $T$ such that
the trajectories emanating from these states always lie within 
$I(T)$, and for states that are far enough, they satisfy LIA formula $G$.
\begin{definition}
Define $\mathrm{SAT}^\infty (T, G(\vec{x})) = \{\vec{y} \in I(T): \exists \vec{y}_1, \vec{y}_2, \ldots \in I(T).
     \vec{y} \rightarrow_T \vec{y}_1 \rightarrow_T \vec{y}_2 \cdots \land \exists K. \forall k > K. G(\vec{y}_k)\}$.
\end{definition}

We then point out that the set of states returned by
Algorithm~\ref{alg:end-to-end-mp} is characterized by $\mathrm{SAT}^\infty $.
\begin{lemma}
  \label{lem:semantics-of-char-sequences}
  Let $T$ be a $\QDLTS$ such that $S_T = \mathbb{Q}^n$.
  Suppose we compute its integer-spectrum restriction, where 
   $M_T$ is the matrix representation of $T|_{I(T)}$, and
  $P_T$ is the matrix representation of some linear projection of $S_T$ onto $I(T)$ (w.l.o.g. we assume $P_T$ has integer entries).
  Let $D_T$ be a matrix such that $D_T \vec{y} = 0$ iff $\vec{y} \in I(T)$.
  Let $G(\vec{x})$ be any LIA formula over $\vec{x}$, and
  $G'(\vec{w}) = \exists \vec{x}. G(\vec{x}) \land D_T \vec{x} = 0 \land \vec{w} = P_T \vec{x}$.
  Let $(H_0(\vec{w}), \dots, H_p(\vec{w}))^\omega = \chi(G'(\vec{w}),M_T)$.
  Then the following are equivalent:
  \begin{enumerate}
    \item $H_0(P_T \vec{y}) \land \dots \land H_p(P_T \vec{y}) \land D_T \vec{y} = 0$ holds, and
    \item $\vec{y} \in \mathrm{SAT}^\infty (T, G(\vec{x}))$.
  \end{enumerate}
\end{lemma}
\begin{proof}
  First, we prove that if $\vec{y}$ satisfies $\bigwedge_i H_i(P_T \vec{y})$ and $D_T \vec{y} = 0$, then
  $\vec{y} \in \mathrm{SAT}^\infty (T, G(\vec{x}))$.
  Since $D_T \vec{y} = 0 $, we have $\vec{y} \in I(T)$.
  Since $I(T)$ is a subspace of $\dom^\omega(T)$,
  there is an infinite trajectory $\vec{y} \rightarrow_T \vec{y}_1 \rightarrow_T \vec{y_2} \rightarrow_T \ldots$.
  Since $M_T$ is a representation of $T|_{I(T)}$ and $P_T$ is a representation of a projection from $S_T$ onto $I(T)$, 
  for any $k$  we have that $M_T^kP_T\vec{y} = P_T\vec{y}_k$.
  By Theorem~\ref{thm:periodic-truth-value-of-guard-under-linear-map}
  and the assumption that $\bigwedge_i H_i(P_T\vec{y})$ holds,
  there exists $K$ such that $\forall k \geq K$, $G'(P\vec{y}_k)$ holds.
  It follows that for $k \geq K$, there is some $\vec{z}_k$ such that
  $G(\vec{z}_k) \land D_T \vec{z}_k = 0 \land P_T\vec{y}_k = P_T \vec{z}_k$ holds.  Since $D_T\vec{z} = 0$, we have $\vec{z} \in I(T)$.  Since $P_T$ is injective on $I(T)$, and $P_T\vec{x}_k = P_T\vec{y}_k$, we must have $\vec{x}_k = \vec{y}_k$, and thus $G(\vec{y}_k)$ holds for all $k \geq K$, and finally $\vec{y} \in \mathrm{SAT}^\infty (T, G(\vec{x}))$.

  Next, we prove that given $\vec{y} \in \mathrm{SAT}^\infty (T, G(\vec{y}))$, 
  it must satisfy $\bigwedge_i H_i(P_T \vec{y})$ and $D_T \vec{y} = 0$.
  Since $\vec{y} \in I(T)$, we have $D_T\vec{y} = 0$.
  Since $\vec{y} \in \mathrm{SAT}^\infty (T, G(\vec{y}))$, there is an infinite trajectory 
  $\vec{y} \rightarrow_T \vec{y}_1 \rightarrow_T \vec{y_2} \rightarrow_T \ldots$ and a $K \in \mathbb{N}$ such that
  $G(\vec{y}_k)$ holds for all $k \geq K$.  Letting $\vec{w}_k = P_T\vec{y}_k$ for all $k$, we have that for all $k \geq K$
  $G(\vec{y}_k) \land D_T\vec{y}_k = 0 \land \vec{w}_k = P_T\vec{y}_k$, and thus $G'(\vec{w}_k)$ holds.
  Since $M_T^kP_T\vec{y} = P_T\vec{y}_k = \vec{w}_k$ for all $k$,
  we have $\bigwedge_i H_i(P_T \vec{y})$ by Theorem~\ref{thm:periodic-truth-value-of-guard-under-linear-map}.
\end{proof}

\printProofs[monotonicity]
